\documentclass{article}

\usepackage[english]{babel}

\usepackage[letterpaper,top=2cm,bottom=2cm,left=2cm,right=2cm,marginparwidth=1.75cm]{geometry}

\usepackage{amsmath}
\usepackage{graphicx}
\usepackage[colorlinks=true, allcolors=blue]{hyperref}
\usepackage{indentfirst} 
\usepackage{booktabs}
\usepackage{amsfonts}
\usepackage{amsmath}
\usepackage{multirow}
\usepackage{amsthm}
\usepackage{diagbox}
\usepackage{bm}
\usepackage{color}
\usepackage{graphicx}
\usepackage{natbib}
\usepackage{enumitem}
\usepackage{algorithm}
\usepackage{algorithmic}
\usepackage{mathtools}
\usepackage{thmtools} 
\usepackage{thm-restate}
\usepackage{xcolor}         % colors
\definecolor{pku-red}{RGB}{139,0,18}
\usepackage[capitalize,noabbrev,nameinlink]{cleveref}
\usepackage[textsize=tiny]{todonotes}
\theoremstyle{plain}
\newtheorem{theorem}{Theorem}[section]

\newtheorem{lemma}[theorem]{Lemma}
\newtheorem{corollary}[theorem]{Corollary}
\newtheorem{example}[theorem]{Example}

\theoremstyle{definition}
\newtheorem{definition}[theorem]{Definition}

\theoremstyle{remark}

\newcommand{\EE}{\mathbb{E}}
\newcommand{\PP}{\mathbb{P}}
\newcommand{\RR}{\mathbb{R}}

\newcommand{\NN}{\mathbb{N}}

\newcommand{\Ff}{\mathcal{F}}

\newcommand{\Nn}{\mathcal{N}}
\newcommand{\Aa}{\mathcal{A}}
\newcommand{\Bb}{\mathcal{B}}

\newcommand{\Ee}{\mathcal{E}}
\newcommand{\Dd}{\mathcal{D}}

\newcommand{\Uu}{\mathcal{U}}
\newcommand{\Xx}{\mathcal{X}}

\newcommand{\Yy}{\mathcal{Y}}
\newcommand{\Rr}{\mathcal{R}}
\newcommand{\Qq}{\mathcal{Q}}
\newcommand{\Pp}{\mathcal{P}}
\newcommand{\Oo}{\mathcal{O}}
\newcommand{\Gg}{\mathcal{G}}

\newcommand{\CE}{\mathrm{CE}}
\newcommand{\NE}{\mathrm{NE}}

\newcommand{\norm}[1]{{\|#1\|}}
\newcommand{\fNE}{f^{\mathrm{NE}}}
\newcommand{\gNE}{g^{\mathrm{NE}}}
\newcommand{\FfNE}{\mathcal{F}^{\mathrm{NE}}}
\newcommand{\fCE}{f^{\mathrm{CE}}}

\newcommand{\FfCE}{\mathcal{F}^{\mathrm{CE}}}
\newcommand{\fCCE}{f^{\mathrm{CCE}}}
\newcommand{\FfCCE}{\mathcal{F}^{\mathrm{CCE}}}
\newcommand{\fcCE}{f^{\mathrm{(C)CE}}}
\newcommand{\gcCE}{g^{\mathrm{(C)CE}}}
\newcommand{\FfcCE}{\mathcal{F}^{\mathrm{(C)CE}}}
\newcommand{\SWR}{\mathrm{SWR}_{\mathrm{N,M}}}
\newcommand{\SW}{\mathrm{SW}}

\makeatletter

\title{Are Equivariant Equilibrium Approximators Beneficial?}
\author{
	\textbf{Zhijian Duan}$^{1}$, 
	\textbf{Yunxuan Ma}$^{1}$,
	\textbf{Xiaotie Deng}$^{1,2}$ 
	\\
	$^{1}$Center on Frontiers of Computing Studies, Peking University\\
	$^{2}$Center for Multi-Agent Research, Institute for AI, Peking University
	\\
	\texttt{\{zjduan,charmingmyx,xiaotie\}@pku.edu.cn}
}
\date{}

\begin{document}
	\maketitle
	
	\begin{abstract}
	Recently, remarkable progress has been made by approximating Nash equilibrium (NE), correlated equilibrium (CE), and coarse correlated equilibrium (CCE) through function approximation that trains a neural network to predict equilibria from game representations. Furthermore, equivariant architectures are widely adopted in designing such equilibrium approximators in normal-form games. In this paper, we theoretically characterize benefits and limitations of equivariant equilibrium approximators. For the benefits, we show that they enjoy better generalizability than general ones and can achieve better approximations when the payoff distribution is permutation-invariant. For the limitations, we discuss their drawbacks in terms of equilibrium selection and social welfare. Together, our results help to understand the role of equivariance in equilibrium approximators.
\end{abstract}
	
	% \pagestyle{headings}
% \pagenumbering{arabic}
\section{Introduction}
The equivariant equilibrium property states that,
given a Nash Equilibrium (NE) solution of a game, the permuted solution is also an NE for the game whose actions of representation are permuted in the same way.
% permuting its actions in the input data representation will result in a Nash equilibrium solution that is permuted in the same way.
The same property also holds in correlated equilibrium (CE) and coarse correlated
equilibrium (CCE), as well as the approximate solutions for all three solution concepts. 

In this paper, we are interested in understanding the equivariant equilibrium property in designing neural networks 
that predict equilibria from game payoffs, following such recent approaches in designing equivariant equilibrium approximators~\citep{feng2021neural,marris2022turbocharging} in normal-form games. 
Informally, such equivariant approximators keep the same permutation of the output strategies (represented as vectors or tensors) when the input game representations (e.g., the game payoff tensors) are permuted
\footnote{We will provide a formal definition of equivariance equilibrium approximators in \cref{sec:symmetric}}.
While equivariant approximators achieved empirical success, little work has theoretically discussed whether they are beneficial.

\subsection{Our Contributions}
% In this paper, we theoretically characterize the benefits and limitations of equivariant NE, CE and CCE approximators. 
% We evaluate the approximators using equilibrium approximation, which is defined as the maximum exploitability~\citep{lockhart2019computing,goktas2022exploitability} over all players (See \cref{def:E}).
% We use orbit averaging~\citep{elesedy2021provably} to construct a symmetric subspace of the approximator function class.

We theoretically characterize benefits and limitations of equivariant NE, CE and CCE approximators. 
For the benefits, we first show that equivariant approximators enjoy better generalizability, where we evaluate the approximators using the maximum exploitability~\citep{lockhart2019computing,goktas2022exploitability} over all players.
To get such a result, we derive the generalization bounds and the sample complexities of the NE, CE, and CCE approximators:
The generalization bounds offer confidence intervals on the expected testing approximations based on the empirical training approximations;
The sample complexities describe how many training samples the equilibrium approximators need to achieve desirable generalizability.
The generalization bounds and sample complexities include the covering numbers~\citep{shalev2014understanding}, which measure the representativeness of the approximators' function classes.
Afterward, we prove that the equivariant approximators have lower covering numbers than the general models, therefore have lower generalization bounds and sample complexities. 
We then show that the equivariant approximators can achieve better approximation when the payoff distribution is permutation-invariant.

As for the limitations, we find the equivariant approximators unable to find all the equilibria of some normal-form games. 
Such a result is caused by the limited representativeness of the equivariant approximators' function class.
Besides, we find that the equivariant NE approximator may lose social welfare.
Specifically, in an example we constructed, while the maximum NE social welfare is large, the maximum social welfare of NEs that the equivariant NE approximators could find can be arbitrary close to zero. 
Such a negative result inspires us to balance generalizability and social welfare when we design the approximators' architectures.

\subsection{Further Related Work}
% NE always exists in normal-form games~\citep{nash1950equilibrium}, so as CE and CCE since NE $\subseteq$ CE $\subseteq$ CCE.
% \label{sec:related_work}
% However, finding an NE is PPAD-complete even for $2$-player games~\citep{chen2009settling} and $3$-player games~\citep{daskalakis2009complexity}.
% Such negative results lead to increased attention on developing algorithms to approximate NE~\citep{TS0.3393,DFM1/3}.

Solving (approximate) NE, CE, and CCE for a single game are well studied~\citep{fudenberg1998theory,cesa2006prediction}.
However, many similar games often need to be solved~\citep{harris2023metalearning} , both in practice and in some multi-agent learning algorithms~\citep{marris2021multi,liu2022simplex}.
For instance, in repeated traffic routing games~\citep{sessa2020contextual}, the payoffs of games depend on the capacity of the underlying network, which can vary with time and weather conditions.
In repeated sponsored search auctions, advertisers value different keywords based on the current marketing environment~\citep{nekipelov2015econometrics}.
In many multi-agent learning algorithms such as Nash Q-learning~\citep{hu2003nash}, Correlated-Q learning~\citep{greenwald2003correlated}, V-learning~\citep{jin2021v} and  PSRO~\citep{lanctot2017unified}, an NE, CE or CCE of a normal-form game need to be solved in every update step.

In these settings, it is preferred to accelerate the speed of game solving by function approximation:
\citet{marris2022turbocharging} introduces a neural equilibrium approximator to approximate CE and CCE for $n$-player normal-form games;
\citet{feng2021neural} proposes a neural NE approximator in PSRO~\citep{lanctot2017unified};
\citet{wu2022using} designs a CNN-based NE approximator for zero-sum bimatrix games.
Differentiable approximators have also been developed to learn QREs~\citep{ling2018game}, NE in chance-constrained games~\citep{wu2023ccgnet}, and opponent's strategy~\citep{hartford2016deep}.

Equivariance is an ideal property of the equilibrium approximator.
We will discuss the literates of equivariant approximators after formally defining them in \cref{sec:symmetric}.

\subsection{Organization}
The rest of our paper is organized as follows: 
In \cref{sec:preliminary} we introduce the preliminary of game theory and equilibrium approximators.
In \cref{sec:symmetric} we formally define the equivariance of equilibrium approximators.
We present our theoretical analysis of benefits in \cref{sec:benefits} and limitations in \cref{sec:limitation}.
% We discuss the related works in \cref{sec:related_work} and conclude in \cref{sec:conclusion}.
We conclude and point out the future work in \cref{sec:conclusion}.

\section{Preliminary}
\label{sec:preliminary}

In this section, we introduce the preliminary and notations of our paper.
We also provide a brief introduction to equilibrium approximators.

\subsection{Game Theory}

\paragraph{Normal-Form Game}
Let a normal-form game with joint payoff $u$ be $\Gamma_u = (n, \Aa, u)$, in which
\begin{itemize}
    \item 
    $n \in \NN_{\ge 2}$ is the number of players. Each player is represented by the index $i \in [n]$. 
    
    \item 
    $\Aa = \times_{i\in [n]}\Aa_i$ is the product action space of all players, where $\Aa_i = \{1, 2, \dots, m_i\}$.
    For player $i\in [n]$, let $a_i \in \Aa_i$ be a specific action of $i$ (An action is also referred to as a pure strategy).
    A joint action $a = (a_1, a_2, \dots, a_n) \in \Aa$ represents one play of the game in which the player $i$ takes action $a_i$.
    The action space $\Aa$ is a Cartesian product that contains all possible joint actions. 
    We have $|\Aa| = \prod_{i\in [n]}|\Aa_i| = \prod_{i\in [n]}m_i$.
    
    \item 
    $u = (u_i)_{i\in [n]}$ is the joint payoff or utility of the game.
    $u_i$ is an $n$-dimensional tensor (or matrix if $n=2$) describing player $i$'s payoff on each joint action.
    In our paper, following previous literatures~\citep{TS0.3393,DFM1/3}, we normalize all the elements of payoff into $[0, 1]$.
\end{itemize}

A joint (mixed) strategy is a distribution over $\Aa$.
Let $\sigma = \times_{i\in [n]}\sigma_i$ be a product strategy and $\pi \in \Delta\Aa$ be a joint (possibly correlated) strategy.
Denote $\pi_i$ as the marginal strategy of player $i$ in $\pi$. 
The expected utility of player $i$ under $\pi$ is 
\begin{equation*}
    u_i(\pi) = \mathbb{E}_{a \sim \pi}[u_i(a)] = \sum_{a\in \Aa} \pi(a) u_i(a).
\end{equation*}
Besides, on behalf of player $i$, the other players' joint strategy is denoted as $\pi_{-i}$, so as $a_{-i}$ and $\sigma_{-i}$.

\paragraph{Nash Equilibrium (NE)}
% Nash equilibrium (NE), correlated equilibrium (CE), and coarse correlated equilibrium (CCE) are the most important solution concepts in game theory.
% In these equilibria, no player can receive a higher payoff by disobeying her strategy.
We say a product strategy $\sigma^* = (\sigma^*_1, \sigma^*_2, \dots, \sigma^*_n)$ is a NE if each player's strategy is the best response given the strategies of others, i.e.,
\begin{equation}
\tag{NE}
    u_i(\sigma_i,\sigma^*_{-i}) \le u_i(\sigma^*_i, \sigma^*_{-i}),~\forall i\in[n],\sigma_i\in\Delta\Aa_i.
\end{equation} 
Computing NE for even general $2$-player or $3$-player games is PPAD-hard~\citep{chen2009settling, daskalakis2009complexity}, which leads to research on approximate solutions.
For arbitrary $\epsilon > 0$, we say a product strategy $\hat{\sigma}$ is an \emph{$\epsilon$-approximate Nash equilibrium} ($\epsilon$-NE) if no one can achieve more than $\epsilon$ utility gain by deviating from her current strategy.
Formally, 
\begin{equation}
\tag{$\epsilon$-NE}
u_i(\sigma_i, \hat{\sigma}_{-i}) \le u_i(\hat{\sigma}_i, \hat{\sigma}_{-i}) + \epsilon,~\forall i\in[n],\sigma_i\in\Delta\Aa_i.
\end{equation}
The definition of $\epsilon$-NE reflects the idea that players might not be willing to deviate from their strategies when the amount of utility they could gain by doing so is tiny (not more than $\epsilon$).

\paragraph{Coarse Correlated Equilibrium (CCE)}
We say a joint (possibly correlated) strategy $\pi^*$ is a CCE if no player can receive a higher payoff by acting independently, i.e., 
\begin{equation}
\tag{CCE}
    u_i(\sigma_i, \pi^*_{-i}) \le u_i(\pi^*),~\forall i\in[n],\sigma_i\in\Delta\Aa_i,
\end{equation}
and we say $\hat{\pi}$ is an $\epsilon$-approximate coarse correlated equilibrium ($\epsilon$-CCE) for $\epsilon>0$ if
\begin{equation}
    \tag{$\epsilon$-CCE}
    u_i(\sigma_i, \hat\pi_{-i}) \le u_i(\hat\pi) + \epsilon,~\forall i\in [n],\sigma_i\in\Delta\Aa_i,
\end{equation}
The difference between NE and CCE is that in an NE, players execute their strategy individually in a decentralized way, while in a CCE, players' strategies are possibly correlated.
A standard technique to correlate the strategy is sending each player a signal from a centralized controller ~\citep{shoham2008multiagent}.

\paragraph{Correlated Equilibrium (CE)}
CE is similar to CCE, except that in a CE, each player can observe her recommended action before she acts.
Thus, player $i$ deviates her strategy through \emph{strategy modification} $\phi_i: \Aa_i \to \Aa_i$.
$\phi_i$ maps actions in $\Aa_i$ to possibly different actions in $\Aa_i$.
Based on strategy modification, we say a joint (possibly correlated) strategy $\pi^*$ is a CE if 
\begin{equation}
    \tag{CE}
    \sum_{a \in \Aa} \pi^*(a) u_i(\phi_i(a_i), a_{-i}) \le u_i(\pi^*),~\forall i, \forall \phi_i,
\end{equation}
and a joint strategy $\hat\pi$ is an $\epsilon$-approximate correlated equilibrium ($\epsilon$-CE) for $\epsilon>0$ if
\begin{equation}
    \tag{$\epsilon$-CE}
    \sum_{a \in \Aa} \hat\pi(a) u_i(\phi_i(a_i), a_{-i}) \le u_i(\hat\pi)+\epsilon,~\forall i, \forall \phi_i,
\end{equation}

Note that for a finite $n$-player normal-form game, at least one NE, CE, and CCE must exist. 
This is because NE always exists~\citep{nash1950equilibrium} and NE $\subseteq$ CE $\subseteq$ CCE.

\paragraph{Equilibrium Approximation}
To evaluate the quality of a joint strategy to approximate an equilibrium, we define approximation based on exploitability~\citep{lockhart2019computing,goktas2022exploitability}.
\begin{definition}[Exploitability and Approximation]
    \label{def:E}
    Given a joint strategy $\pi$, the \emph{exploitability} (or regret) $\Ee_i(\pi, u)$ of player $i$ is the maximum payoff gain of $i$ by deviating from her current strategy, i.e., 
    \begin{equation*}
        \begin{aligned}
            \Ee_i(\pi, u) &\coloneqq \max_{\sigma'_i}u_i(\sigma'_i, \pi_{-i}) - u_i(\pi) = \max_{a'_i}u_i(a'_i, \pi_{-i}) - u_i(\pi)
        \end{aligned}
    \end{equation*}
    and the exploitability under strategy modification $\Ee_i^\CE(\pi, u)$ of player $i$ is the maximum payoff gain of $i$ by deviating through strategy modification, i.e.,
    \begin{equation*}
        \Ee_i^\CE(\pi, u) \coloneqq \max_{\phi_i}\sum_{a \in \Aa} \pi(a) u_i(\phi_i(a_i), a_{-i}) - u_i(\pi).
    \end{equation*}
    The \emph{equilibrium approximation} is defined as the maximum exploitability over all players    
    \footnote{
        A similar metric of equilibrium approximation is called Nikaido-Isoda function~\citep{nikaido1955note} or \textsc{NashConv}~\citep{lockhart2019computing}, which is the sum of exploitability over all players, i.e., $\sum_{i\in[n]}\Ee_i(\pi, u)$.
    }, i.e.,
    \begin{equation*}
        \Ee(\pi, u) \coloneqq 
        \begin{cases}
        \max_{i\in[n]}\Ee_i(\pi, u)&,\text{for NE and CCE}
        \\
        \max_{i\in[n]}\Ee^\CE_i(\pi, u)&,\text{for CE}
    \end{cases}
    \end{equation*}
\end{definition}
Based on approximation, we can restate the definition of solution concepts.
A product strategy $\sigma$ is an NE of game $\Gamma_u$ if $\Ee(\sigma, u) = 0$ and is an $\epsilon$-NE if $\Ee(\sigma, u) \le \epsilon$. 
A joint strategy $\pi$ is a (C)CE of $\Gamma_u$ if $\Ee(\pi, u) = 0$ and is an $\epsilon$-(C)CE if $\Ee(\pi, u) \le \epsilon$.

\subsection{Equilibrium Approximator}
\label{sec:approximator}

The equilibrium approximators, including NE, CE, and CCE approximators, aim to predict solution concepts from game representations.
In our paper, we fix the number of players $n$ and action space $\Aa$.
We denote by $\Uu$ the set of all the possible game payoffs.
The NE approximator $\fNE: \Uu \to \times_{i\in [n]}\Delta\Aa_i$ maps a game payoff to a product strategy, where $\fNE(u)_i \in \Delta\Aa_i$ is player $i$'s predicted strategy.
We define $\FfNE$ as the function class of the NE approximator.
Similarly, we define (C)CE approximator as $\fcCE: \Uu \to \Delta\Aa$ and (C)CE approximator class as $\FfcCE$.

\begin{algorithm}[t]
	\caption{Example: learning NE/CCE approximator via minibatch SGD}
	\label{alg:training}
	\begin{algorithmic}[1]
		\STATE {\bfseries Input:} Training set $S$
		\STATE {\bfseries Parameters:} Number of iterations $T > 0$, batch size $B > 0$, learning rate $\eta > 0$, initial parameters $w_0 \in \RR^d$ of the approximator model.
		% \vskip 5pt
		\FOR{$t~=~0$ \textbf{to} $T$}
		\STATE Receive minibatch ${S}_t \,=\, \{u^{(1)}, \ldots, u^{(B)}\} \subset S$ 
		
		\STATE Compute the empirical average approximation of $S_t$:
		\STATE ~~~~$L_{S_t}(f^{w_t}) \gets \frac{1}{B}\sum_{i=1}^B \Ee(f^{w_t}(u^{(i)}), u^{(i)}) $
		
		\STATE Update model parameters:
		\STATE ~~~~$w_{t+1} \gets w_t - \eta\nabla_{w_t} L_{S_t}(f^{w_t}) $
		\ENDFOR
	\end{algorithmic}
\end{algorithm}

% An equilibrium approximator can be learned through machine learning paradigms based on empirical data generated from a (possibly unknown) underlying payoff distribution $\Dd$.
% % We will provide the generalization bound of equilibrium approximators in \cref{sec:benefits}.
% % The bound describes the approximation difference of the equilibrium approximator in training and testing, assuming that the training and testing data distribution $\Dd$ are the same.
% Straightforwardly, the NE and CCE approximators can be learned by minibatch stochastic gradient descent (SGD)~\citep{bottou2012stochastic} on equilibrium approximation. 
% % Specifically, denoting $w \in \RR^d$ as the $d$-dimensional parameter of the approximator, such as the weights of the neural network.
% % We can optimize $w$ by the standard minibatch SGD algorithm on equilibrium approximation (See \cref{alg:training}).
% The algorithm is feasible since the equilibrium approximation is differentiable almost everywhere.
% % The non-differentiable zero-measure points appear when the maximum operations in $\Ee(\pi, u)$ have multiple choices.
% % We can set one of them according to any tie-breaking rule as the outcome to compute the corresponding gradient.
An equilibrium approximator can be learned through machine learning paradigms based on empirical data.
For instance, \citet{feng2021neural} learn the NE approximator using the game payoffs generated in the learning process of PSRO, and optimize the approximator by gradient descent in exploitability. 
\citet{marris2022turbocharging} learn the CE and CCE approximators using the games i.i.d. generated from a manually designed distribution, and optimize the approximators using maximum welfare minimum relative entropy loss.
Such a loss balances the equilibrium approximation, the social welfare, and the relative entropy of the joint strategy. 
% The loss can be optimized by solving the dual problem, and the approximators are learned through self-supervised learning.
Additionally, another simple way to learn NE or CCE equilibrium approximator is to apply minibatch stochastic gradient descent (SGD) on approximation.  
Specifically, we denote $w \in \RR^d$ as the $d$-dimensional parameter of the approximator, such as the weights of the neural network.
We can optimize $w$ by the standard minibatch SGD algorithm on approximation (See \cref{alg:training}).

\section{Equivariant Equilibrium Approximator}
\label{sec:symmetric}
In this section, we introduce the equivariance of the equilibrium approximators and the way how we apply orbit averaging~\citep{elesedy2021provably} to construct equivariant approximators. 
Equivariant approximator has been developed in many literatures~\citep{hartford2016deep,feng2021neural,marris2022turbocharging,wu2022using}, which we will discuss latter. 

We first define the permutation of a game.
Let $\rho_i: \Aa_i \to \Aa_i$ be a permutation function of player $i$, which is a bijection from $\Aa_i$ to $\Aa_i$ itself. 
Define $\Gg_i \ni \rho_i$ as the class of permutation function for player $i$, which forms an Abelian group.

\begin{definition}[Permutation of a game]
For a normal-form game $\Gamma_u = (n, u, \Aa)$, we define the $\rho_i$-permutation of payoff tensor $u$ as $\rho_i u = (\rho_i u_j)_{j\in [n]}$, in which
\begin{equation*}
    (\rho_i u_j)(a_i, a_{-i}) = u_j(\rho_i^{-1}(a_i), a_{-i}),~\forall a \in \Aa.
\end{equation*}
We also define the $\rho_i$-permutation of joint strategy $\pi$ as $\rho_i\pi$, where
\begin{equation*}
    (\rho_i\pi)(a_i, a_{-i}) = \pi(\rho_i^{-1}(a_i), a_{-i}),~\forall a \in \Aa,
\end{equation*}
and the $\rho_i$-permutation of product strategy $\sigma$ as $\rho_i\sigma = (\rho_i\sigma_j)_{j\in[n]}$, where
\begin{equation*}
    \forall a_j \in \Aa_j, \rho_i\sigma_j(a_j) =
    \begin{cases}
        \sigma_j(a_j) &,\text{if }j\ne i
        \\
        \sigma_i(\rho_i^{-1} a_i)  &,\text{if }j = i\\
    \end{cases}
\end{equation*}
\end{definition}

\paragraph{Equivariant NE Approximator}
Considering the relationship of $\rho_i$-permuted game and $\rho_i$-permuted product strategy, we have the following result for NE: 
\begin{restatable}{lemma}{lemPermNE}
% \begin{lemma}
    \label{lem:perm:NE}
    In a normal-form game $\Gamma_u = (n, u, \Aa)$, for arbitrary player $i \in [n]$ and any ($\epsilon$-)NE strategy $\sigma = (\sigma_i, \sigma_{-i})$, $\rho_i\sigma = (\rho_i\sigma_i, \sigma_{-i})$ is also an ($\epsilon$-)NE for the $\rho_i$-permuted game $\Gamma_{\rho_i u}$.
% \end{lemma}
\end{restatable}

\autoref{lem:perm:NE} tells the unimportance of action order with respect to NE approximation.  
Based on it, we can summarize two ideal properties for NE approximators, which are defined as follows:

\begin{definition}[Player-Permutation-Equivariance, PPE]
    \label{def:PPE}
    We say an NE approximator $\fNE$ satisfies \emph{player $i$ permutation-equivariant ($i$-PE)} if for arbitrary permutation function $\rho_i \in \Gg_i$ we have
    \begin{equation}
    \label{eq:iPE}
    \tag{$i$-PE}
        \fNE(\rho_i u)_i = \rho_i \fNE(u)_i,
    \end{equation}
    Moreover, we say $\fNE$ is \emph{player-permutation-equivariant (PPE)} if $\fNE$ satisfies \ref{eq:iPE} for all player $i\in [n]$.
\end{definition}

\begin{definition}[Opponent-Permutation-Invariance, OPI]
    We say an NE approximator $\fNE$ is \emph{opponent $i$ permutation-invariant ($i$-PI)} if for any other player $j \in [n] - \{i\}$ and arbitrary permutation function $\rho_i \in \Gg_i$ we have
    \begin{equation}
    \label{eq:iPI}
    \tag{$i$-PI}
        \fNE(\rho_i u)_j = \fNE(u)_j, \forall j\ne i
    \end{equation}
    and we say $\fNE$ is \emph{opponent-permutation-invariant (OPI)} if $\fNE$ satisfies \ref{eq:iPI} for all player $i\in [n]$.
\end{definition}

\paragraph{Equivariant (C)CE approximator}
Considering the relationship of $\rho_i$-permuted game and $\rho_i$-permuted joint strategy, we have a similar result for CE and CCE: 
\begin{restatable}{lemma}{lemPermCCE}
% \begin{lemma}
    \label{lem:perm:CCE}
    In a normal-form game $\Gamma_u = (n, u, \Aa)$, for an arbitrary player $i \in [n]$ and any ($\varepsilon$-)CE or ($\epsilon$-)CCE strategy $\pi$, $\rho_i\pi$ is also an ($\varepsilon$-)CE or ($\epsilon$-)CCE for the $\rho_i$-permuted game $\Gamma_{\rho_i u}$.
% \end{lemma}
\end{restatable}

Inspired by \cref{lem:perm:CCE}, we can also summarize an ideal property for CE and CCE approximators defined as follows.
\begin{definition}[Permutation-Equivariance,PE]
    We say an (C)CE approximator $f^{\text{(C)CE}}$ is \emph{player $i$ permutation-equivariant ($i$-PE)} if for permutation function $\rho_i \in \Gg_i$ we have
    \begin{equation*}
        \fcCE(\rho_i u) = \rho_i \fcCE(u),
    \end{equation*}
    and we say $f^{\text{(C)CE}}$ is \emph{permutation-equivariant (PE)} if $f^{\text{(C)CE}}$ satisfies $i$-PE for all player $i\in [n]$.
\end{definition}

\paragraph{Equivariant Approximators in Literature}
For two-player games, \citet{feng2021neural} propose an MLP-based NE approximator that satisfies both PPE and OPI for zero-sum games.
Additionally, they also design a Conv$1$d-based NE approximator that satisfies PPE only. 
\citet{hartford2016deep} give a PPE approximator to predict players' strategies.
% \citet{wu2022using} designs a CNN-based NE approximator with both PPE and OPI.
The traditional algorithms \citet{TS0.3393} and \citet{DFM1/3}, which approximate NE by optimization, are also PPE and OPI to payoff and the initial strategies.
For $n$-player general games, \citet{marris2022turbocharging} provide a permutation-equivariant approximator to approximate CE and CCE.
Equivariant architectures are also adopted in optimal auction design~\citep{rahme2021permutation,duan2022context,ivanov2022optimaler}, and \citet{qin2022benefits} theoretically characterize the benefits of permutation-equivariant in auction mechanisms.
We follow the rough idea of \citet{qin2022benefits} when we analyze the benefits of equivariant equilibrium approximators.

\subsection{Orbit Averaging}
Orbit averaging is a well-known method to enforce equivariance or invariance for a function~\citep{schulz1994constructing}.
It averages the inputs of a function over the orbit of a group (e.g., the permutation group in our paper).

\paragraph{Orbit Averaging for NE Approximator}
For an NE approximator $\fNE$ and any player $i \in [n]$, we can construct a \ref{eq:iPI} or \ref{eq:iPE} NE approximator by averaging with respect to all the permutations of player $i$. 
Specifically, we construct an \ref{eq:iPI} NE approximator by operator $\Oo_i$ with  
\begin{equation*}
    (\Oo_i \fNE)(u)_j = 
    \begin{cases}
        \fNE(u)_i &,\text{if }j = i 
        \\
        \frac{1}{|\Aa_i|!} \sum_{\rho_i \in \Gg_i}\fNE(\rho_i u)_j &,\text{otherwise}
    \end{cases}
\end{equation*}
and we construct an \ref{eq:iPE} NE approximator by operator $\Pp_i$ with: 
\begin{equation*}
    (\Pp_i \fNE)(u)_j = 
    \begin{cases}
        \frac{1}{|\Aa_i|!} \sum_{\rho_i \in \Gg_i} \rho_i^{-1}\fNE(\rho_i u)_i &,\text{if }j = i
        \\
        \fNE(u)_j &,\text{otherwise}
    \end{cases}
\end{equation*}

% We also define the third operator $\Qq_i = \Oo_i \circ \Pp_i$ as a composition to construct an NE approximator with both $i$-PE and \ref{eq:iPI}:
% \begin{equation*}
%     (\Qq_i \fNE)(u)_j = 
%     \begin{cases}
%         \frac{1}{|\Aa_i|!} \sum_{\rho_i \in \Gg_i} \rho_i^{-1}\fNE(\rho_i u)_i &,\text{if }j=i 
%         \\
%         \frac{1}{|\Aa_i|!} \sum_{\rho_i \in \Gg_i}\fNE(\rho_i u)_j &,\text{otherwise}
%     \end{cases}
% \end{equation*}

\begin{restatable}{lemma}{lemNEOiPi}
% \begin{lemma}
\label{lem:NE:OiPi}
    $\Oo_i\fNE$ is \ref{eq:iPI} and $\Pp_i\fNE$ is \ref{eq:iPE}. 
    % $\Qq_i\fNE$ satisfies both \ref{eq:iPI} and \ref{eq:iPE}.
    Specially, if $\fNE$ is already \ref{eq:iPI} or \ref{eq:iPE}, then we have $\Oo_i\fNE = \fNE$ or $\Pp_i\fNE = \fNE$, respectively.
% \end{lemma}
\end{restatable}

To construct a PPE or OPI NE approximator, we composite the operators with respect to all players.
Let $\Oo = \Oo_1 \circ \Oo_2 \circ \dots \circ \Oo_n$ and $\Pp = \Pp_1 \circ \Pp_2 \circ \dots \circ \Pp_n$, we get the following corollary:
\begin{restatable}{lemma}{lemNEOP}
% \begin{lemma}
\label{lem:NE:OP}
    $\Oo\fNE$ is OPI and $\Pp\fNE$ is PPE. 
    % $\Qq\fNE$ satisfies both OPI and PPE.
    If $\fNE$ is already OPI or PPE, we have $\Oo\fNE = \fNE$ or $\Pp\fNE = \fNE$, respectively. 
% \end{lemma}
\end{restatable}
Furthermore, we can also compose $\Pp \circ \Oo$ to construct a NE approximator with both PPE and OPI.

\paragraph{Orbit Averaging for (C)CE Approximator}
For CE or CCE approximator $f$, we define $\Qq_i$-project for player $i \in [n]$ to construct an $i$-PE approximator, which averages with respect to all the permutations of player $i$.
\begin{equation*}
    (\Qq_i f^{\text{(C)CE}})(u) = \frac{1}{|\Aa_i|!} \sum_{\rho_i \in\Gg_i} \rho_i^{-1} f^{\text{(C)CE}}(\rho_i u)
\end{equation*}
Similarly, we define $\Qq = \Qq_1 \circ \Qq_2 \circ \dots \circ \Qq_n$ as the composite operator.
\begin{restatable}{lemma}{lemCCEQ}
% \begin{lemma}
\label{lem:CCE:Q}
    $\Qq_i f^{\text{(C)CE}}$ is $i$-PE and $\Qq f^{\text{(C)CE}}$ is PE.
    Specifically, If $f^{\text{(C)CE}}$ is already $i$-PE or PE, then we have $\Qq_i f^{\text{(C)CE}} = f^{\text{(C)CE}}$ or $\Qq f^{\text{(C)CE}} = f^{\text{(C)CE}}$, respectively.
% \end{lemma}
\end{restatable}

Combined with \cref{lem:NE:OiPi}, \cref{lem:NE:OP} and \cref{lem:CCE:Q}, we can have the following corollary directly.
\begin{corollary}
\label{cor:idempotent}
$\Oo^2 = \Oo, \Pp^2 = \Pp, \Qq^2 = \Qq$.
\end{corollary}

The benefit of using orbit averaging is shown in the following lemma:
\begin{restatable}{lemma}{lemSubset}
% \begin{lemma}
    \label{lem:subset}
    Denote $\Xx$ as an idempotent operator, i.e. $\Xx^2 = \Xx$ (e.g. $\Oo, \Pp$ or $\Qq$).
    For function class $\Ff$ of NE, CE or CCE approximator, let $\Ff_\Xx$ be any subset of $\Ff$ that is closed under $\Xx$, then $\Xx \Ff_\Xx$ is the largest subset of $\Ff_\Xx$ that is invariant under $\Xx$.
% \end{lemma}
\end{restatable}

According to \cref{lem:NE:OP}, \cref{lem:CCE:Q} and \cref{lem:subset}, $\Oo\FfNE$(or $\Pp\FfNE$/$\Qq\FfcCE$) is the largest subset of $\FfNE$(or $\FfNE$/$\FfcCE$) with the corresponding property (OPI, PPE, or PE) if $\FfNE$(or $\FfNE$/$\FfcCE$) is closed operator under $\Oo$(or $\Pp$/$\Qq$).
The result tells that the orbit averaging operators, while enforcing the operated function to be equivariance or invariance, keep as large capacity of the function class as possible. 
Therefore, we believe that orbit averaging is an ideal approach to constructing equivariant or invariant functions.

% -----

\section{Theoretical Analysis of Benefits}
\label{sec:benefits}

In this section, we theoretically analyze the benefits of equivariant approximators with respect to generalizability and approximation.

\subsection{Benefits for Generalization}
We first derive the generalization bound and sample complexity for general approximator classes, and then we show the benefits of equivariant approximators by applying orbit averaging to the approximators.

The representativeness of an approximator class is measured by the covering numbers~\citep{shalev2014understanding} under $\ell_{\infty}$-distance, which are defined as follows:
\begin{definition}[$\ell_{\infty}$-distance] 
The $\ell_{\infty}$-distance between two equilibrium approximators $f, g$ is:
\begin{equation*}
    \ell_{\infty}(f,g) = \max_{u \in \Uu}\norm{f(u) - g(u)},
\end{equation*}
where we define the distance of two product strategies $\sigma$ and $\sigma'$ as
\begin{equation*}
    \norm{\sigma^1 - \sigma^2} 
    = \max_{i\in[n]} \sum_{a_i \in \Aa_i} |\sigma^1_i(a_i) - \sigma^2_i(a_i)|
\end{equation*}
and the distance of two joint strategy $\pi$ and $\pi'$ as
\begin{equation*}
    \norm{\pi^1 - \pi^2} = 
    \sum_{a \in \Aa} |\pi^1(a) - \pi^2(a)|
\end{equation*}
\end{definition}

\begin{definition}[$r$-covering number]
    For $r > 0$, we say function class $\Ff_r$ $r$-covers another function class $\Ff$ under $\ell_{\infty}$-distance if for all function $f \in \Ff$, there exists $f_r \in \Ff_r$ such that $\norm{f - f_r}_{\infty} \le r$.
    The $r$-covering number $\Nn_{\infty}(\Ff, r)$ of $\Ff$ is the cardinality of the smallest function class $\Ff_r$ that $r$-covers $\Ff$ under $\ell_{\infty}$-distance.
\end{definition}

Based on covering numbers, we provide the generalization bounds of NE, CE and CCE approximators. 
The bounds describe the difference between the expected testing approximation and empirical training approximation.

\begin{restatable}{theorem}{thmGB}[Generalization bound]
% \begin{theorem}
\label{thm:GB}
    For function class $\Ff$ of NE, CE or CCE approximator, with probability at least $1 - \delta$ over draw of the training set $S$ (with size $m$) from payoff distribution $\Dd$, for all approximator $f \in \Ff$ we have
    \begin{equation*}
    \begin{aligned}
    	&\EE_{u\sim\Dd}[\Ee(f(u), u)] - \frac{1}{m}\sum_{u\in S}\Ee(f(u), u) 
    	\le 2\cdot \inf_{r>0}\{ \sqrt{\frac{2\ln\Nn_{\infty}(\Ff, r)}{m}} + Lr \}
    	+ 4\sqrt{\frac{2\ln(4/\delta)}{m}},
    \end{aligned}
    \end{equation*}
    where $L = 2n$ for NE approximator, and $L = 2$ for CE and CCE approximators.
% \end{theorem}
\end{restatable}

To get the theorem, we first show that all three equilibrium approximations are Lipschitz continuous with respect to strategies.
Afterward, we derive the Rademacher complexity~\citep{bartlett2002rademacher} of the expected approximation based on the Lipschitz continuity and covering numbers. 
See \cref{prf:thm:GB} for the detailed proof. 

We can see from \cref{thm:GB} that, with a large enough training set, the generalization gaps of equilibrium approximators go to zero if the covering number $\Nn_\infty(\Ff, r)$ is bounded.
As a result, we can estimate the expected testing performance through the empirical training performance.

We can also derive the sample complexities of equilibrium approximators to achieve the desirable generalizability.
\begin{restatable}{theorem}{thmSC}[Sample complexity]
% \begin{theorem}
\label{thm:SC}
    For $\epsilon, \delta \in (0, 1)$, function class $\Ff$ of NE, CE or CCE approximator and distribution $\Dd$, with probability at least $1 - \delta$ over draw of the training set $S$ with
    \begin{equation*}
    	m \ge \frac{9}{2\epsilon^2}\left(\ln\frac{2}{\delta} + \ln\Nn_{\infty}(\Ff, \frac{\epsilon}{3L})\right)
    \end{equation*}
    games sampled from $\Dd$, $\forall f \in \Ff$ we have 
    \begin{equation*}
        \EE_{u\sim\Dd}[\Ee(f(u), u)] \le \frac{1}{m}\sum_{u\in S}\Ee(f(u), u) + \epsilon,
    \end{equation*}
    where $L = 2n$ for NE approximators, and $L = 2$ for CE and CCE approximators.
% \end{theorem}
\end{restatable}

The proof is based on the Lipschitz continuity of approximation, uniform bound, and concentration inequality.  
See \cref{prf:thm:SC} for details.
\cref{thm:SC} is also called the uniform convergence of function class $\Ff$, which is a sufficient condition for agnostic PAC learnable~\citep{shalev2014understanding}. 

As for the benefits of equivariant approximators for generalizability,
the following result indicates that the projected equilibrium approximators have smaller covering numbers.
\begin{restatable}{theorem}{thmCO}
% \begin{theorem}
\label{thm:CO}
The $\Oo$-projected, $\Pp$-projected and $\Qq$-projected approximator classes have smaller covering numbers, i.e., $\forall r > 0$ we have
\begin{align*}
    \Nn_{\infty}(\Oo\FfNE, r) &\le  \Nn_{\infty}(\FfNE, r), 
    \\
    \Nn_{\infty}(\Pp\FfNE, r) &\le  \Nn_{\infty}(\FfNE, r),  
    \\
    \Nn_{\infty}(\Qq\FfcCE, r) &\le \Nn_{\infty}(\FfcCE, r)
\end{align*}
% \end{theorem}
\end{restatable}
The proof is done by showing all the operators are contraction mappings. See \cref{prf:thm:CO} for details.

Both the generalization bounds in \cref{thm:GB} and the sample complexities in \cref{thm:SC} decrease with the decrease of covering numbers $\Nn_\infty(\Ff, r)$.
Thus, we can see from \cref{thm:CO} that both PPE and OPI can improve the generalizability of NE approximators, and PE can improve the generalizability of CE and CCE approximators.

\subsection{Benefits for Approximation}
We then show the benefits of equivariance for approximation when the payoff distribution is invariant under permutation.
The permutation-invariant distribution holds when the action is anonymous or indifferent or when we pre-train the equilibrium approximators using a manually designed distribution~\citep{marris2022turbocharging}.

\paragraph{(C)CE Approximator}
The following theorem tells the benefit of permutation-equivariance in decreasing the exploitability of (C)CE approximators.
\begin{restatable}{theorem}{thmCCEE}
% \begin{theorem}
\label{thm:CCE:E}
    When the payoff distribution $\Dd$ is invariant under the permutation of payoffs, the $\Qq$-projected (C)CE approximator has a smaller expected equilibrium approximation. Formally, for all $\fcCE \in \FfcCE$ and permutation-invariant distribution $\Dd$, we have
    \begin{equation*}
        \EE_{u \sim \Dd}[\Ee(\Qq \fcCE(u), u)] \le \EE_{u \sim \Dd}[\Ee(\fcCE(u), u)],
    \end{equation*}
% \end{theorem}
\end{restatable}

The proof is done by using the convexity of approximation. See \cref{prf:thm:CCE:E} for details.
We can see from \cref{thm:CCE:E} that, when payoff distribution is invariant under permutation, it is beneficial to use equivariant architecture as the CE or CCE approximators.

\paragraph{NE Approximator}
As for NE approximator, we have similar results.

\begin{restatable}{theorem}{thmNEBiC}
\label{thm:NE:Bi:c}
    For bimatrix constant-sum games, when the payoff distribution $\Dd$ is invariant under the permutation of payoffs, then the $\Xx$-projected ($\Xx \in \{\Oo, \Pp\}$) NE approximator has a smaller expected exploitability. Formally, for all $\fNE \in \FfNE$ and permutation-invariant distribution $\Dd$ for bimatrix constant-sum games, we have
    \begin{equation*}
        \EE_{u \sim \Dd}[\sum_i \Ee_i((\Xx \fNE)(u), u)] \le \EE_{u \sim \Dd}[\sum_i \Ee_i(\fNE(u), u)]
    \end{equation*}
\end{restatable}

\begin{restatable}{theorem}{thmNEP}
\label{thm:NE:P}
    When the payoff distribution $\Dd$ is invariant under the permutation of payoffs, and $\fNE$ satisfies OPI, then the $\Pp$-projected NE approximator has a smaller expected NE approximation. Formally, for all $\fNE \in \FfNE$ that is OPI and permutation-invariant distribution $\Dd$, we have
    \begin{equation*}
        \EE_{u \sim \Dd}[\Ee((\Pp \fNE)(u), u)] \le \EE_{u \sim \Dd}[\Ee(\fNE(u), u)].
    \end{equation*}
\end{restatable}

\begin{restatable}{theorem}{thmNEO}
\label{thm:NE:O}
    For bimatrix games, when the payoff distribution $\Dd$ is invariant under the permutation of payoffs, and $\fNE$ satisfies PPE, then the $\Oo$-projected NE approximator has a smaller expected NE approximation. Formally, for all $\fNE \in \FfNE$ that is PPE and permutation-invariant distribution $\Dd$ of bimatrix games, we have
    \begin{equation*}
        \EE_{u \sim \Dd}[\Ee((\Oo \fNE)(u), u)] \le \EE_{u \sim \Dd}[\Ee(\fNE(u), u)].
    \end{equation*}
\end{restatable}

\cref{thm:NE:P} and \cref{thm:NE:O} tell that PPE and OPI approximators can achieve better approximation than ones with only PPE or OPI.
Meanwhile, we can see from \cref{thm:NE:Bi:c} that for bimatrix constant-sum games (such as zero-sum games), it can be preferred to introduce PPE or OPI to the architectures.

% -----

\section{Theoretical Analysis of Limitations}
\label{sec:limitation}

As we discussed in \cref{sec:benefits}, equivariant approximators enjoy better generalizability and better approximation sometimes.
However, as we will show, they have some limitations regarding equilibrium selection and social welfare.
Such limitations attribute to the limited representativeness caused by equivariance.

\subsection{Equilibrium Selection}
We first show that there may be equilibria points that equivariant approximators will never find.
We illustrate such limitation in permutation-invariant games, which is defined as follows:
\begin{definition}[Permutation-$\rho$-Invariant Game]
\label{def:PERM}
We say a game $\Gamma_u$ is permutation-$\rho$-invariant, where $\rho = \circ_{i\in [n]}\rho_i$, if the payoff $u$ is permutation-invariant with respect to $\rho$. That is, $\rho u = u$.
\end{definition}
Permutation-$\rho$-invariance indicates that one cannot distinguish joint action $a$ from $\rho a$ using only the payoff $u$. 
We'd like to provide an example to show more insight of permutation-$\rho$-invariant games:
\begin{example}
For a $2$-player game $\Gamma_u = (2, u=(u_1, u_2), \Aa = ([m_1], [m_2]))$ , Let $\rho_i=(m_i, m_i-1,\dots,1)$ and $\rho = \rho_1\circ\rho_2$.
If one of the following conditions holds, then $u$ is permutation-$\rho$-invariant:
\begin{enumerate}
    \item $u_1$ and $u_2$ are symmetric and persymmetric (i.e.,  symmetric with respect to the northeast-to-southwest diagonal) squares.
    \item Both $u_1$ and $u_2$ are centrosymmetric, i.e., 
    $u_i(x, y) = u_i(m_1 + 1 - x, m_2 + 1 - y)$ for $i \in \{1, 2\}, x \in [m_1]$ and $y \in [m_2]$. 
\end{enumerate} 
\end{example}

For permutation $\rho = (\circ_{i\in[n]}\rho_i)$ and player $k\in [n]$, we denote the set of non-fixed actions of player $k$ under $\rho_k$ as
\begin{equation*}
    V(\rho_k) \coloneqq \{a_k|a_k\in\Aa_k,\rho_k(a_k)\neq a_k\}.
\end{equation*}
Based on $V(\rho_k)$, we find some equilibria points of permutation-$\rho$-invariant games that any equivariant approximators will never find.

\begin{restatable}{theorem}{thmNEG}
\label{thm:NEG}
For a permutation-$\rho$-invariant game $\Gamma_u$. 
if there is a pure NE $a^* = (a^*_i)_{i\in[n]}$ and at least one player $k\in [n]$ such that $a^*_k \in V(\rho_k)$, then $a^*$ will never be found by any NE approximator with both PPE and OPI.
Besides, $a^*$ (as a pure CE or CCE) will also never be found by any CE or CCE approximator with PE.
\end{restatable}

We illustrate \cref{thm:NEG} by the following example:
\begin{example}
\label{example:NE}
Consider a bimatrix game with identity utility
\begin{equation*}
\label{eq:example:u}
    u=\begin{bmatrix}
    \bm{1,1} & 0,0 \\ 
    0,0 & \bm{1,1}
    \end{bmatrix}   
\end{equation*}
There are two pure NE (bolded in the above matrix) and one mixed NE of $\sigma_1=(0.5, 0.5)$ and $\sigma_2=(0.5, 0.5)$.
Let $\rho_i$ be the unique permute function (except for identity function) of player $i \in [2]$, and $\rho = \rho_1 \circ \rho_2$.
The game is permutation-$\rho$-invariant.

\textbf{Case 1}:
Let $f$ be a permutation-equivariant CE or CCE approximator, and denote $\pi = f(u)$. We have 
\begin{equation*}
    \pi = f(u) \overset{(a)}{=} f(\rho u) \overset{(b)}{=} \rho f(u), 
\end{equation*}
where $(a)$ holds by permutation-$\rho$-invariance of $u$, and $(b)$ holds by PE of $f$.
Thus, we have $\pi_{1,1}=\pi_{2,2} \in [0, \frac{1}{2}]$ and $\pi_{1,2} = \pi_{2,1} \in [0, \frac{1}{2}]$.
As a result, the two pure (C)CEs cannot be found.

\textbf{Case 2}:
Let $f$ be a NE approximator that holds PPE and OPI.
Denote $f(u)=(\sigma_1,\sigma_2)$, where $\sigma_1 = (p_1,1-p_1)$ and $\sigma_2=(p_2,1-p_2)$.
By PPE and OPI of $f$, we have 
% $f(u)_1 = (p_1, 1 - p_1) 
% \overset{(a)}{=} f(\rho_1 \rho_2 u)_1 \overset{(b)}{=} \rho_1 f(\rho_2 u)_1 
% \overset{(c)}{=} \rho_1 f(u)_1
% = (1-p_1,p_1)$,
\begin{equation*}
    \begin{aligned}
        f(u)_1 &= (p_1, 1 - p_1) 
        \overset{(a)}{=} f(\rho_1 \rho_2 u)_1 \overset{(b)}{=} \rho_1 f(\rho_2 u)_1 
        \overset{(c)}{=} \rho_1 f(u)_1
        = (1-p_1,p_1),
    \end{aligned}
\end{equation*}
where $(a)$ holds by permutaion-$\rho$-invariance of $u$, $(b)$ holds by PPE of $f$, and $(c)$ holds by OPI of $f$.
As a result, the only NE that $f$ could find is the mixed NE.
\end{example}

As we can see from the example and \cref{thm:NEG}, the equivariance, while introducing inductive bias to the approximator architecture, is also a strong constraint. 
Such a constraint is why the equivariant approximators cannot find all the equilibria points.

\subsection{Social Welfare}

The social welfare of a joint strategy $\pi$ is defined as the sum of all players' utilities, i.e., 
\begin{equation*}
    \SW(\pi,u) = \sum_{i\in[n]} u_i(\pi).
\end{equation*}
The equilibrium with higher social welfare is usually preferred~\citep{marris2022turbocharging}.

To analyze the social welfare of equivariant approximators, we define the worst social welfare ratio as follows:
\begin{definition}
For any $N, M \ge 2$ and two NE (or CE/CCE) approximator classes $\Ff_1, \Ff_2$ that target on games with number of players $n \le N$ and $|\Aa_i|\le M$, we define the worst social welfare ratio of $\Ff_1$ over $\Ff_2$ as:
\begin{equation*}
\begin{aligned}
    &\SWR(\Ff_1,\Ff_2) \coloneqq \inf_{\Dd}\frac{ \max_{f_1 \in \Ff_1}\EE_{u\sim\Dd} \SW( f_1(u),u)}{ \max_{f_2 \in \Ff_2} \EE_{u\sim\Dd} \SW(f_2(u),u)}
\end{aligned}
\end{equation*}
\end{definition}

$\SWR(\Ff_1, \Ff_2)$ measures the relative representativeness of $\Ff_1$ over $\Ff_2$ in terms of social welfare.
Based on that, we have the following result for equivariant CE and CCE approximator classes:
\begin{restatable}{theorem}{thmNEGRatioCCE}
% \begin{theorem}
\label{thm:NEG:ratio:CCE}
Given $N,M \ge 2$, let $\FfcCE_\mathrm{PE}$ be the function class (target on games with number of players $n \le N$ and $|\Aa_i|\le M$) of all the (C)CE approximators with PE.
Denote by $\FfcCE_\mathrm{general}$ the function class of all the (C)CE approximators.
Then we have
\begin{equation*}
    \begin{aligned}
        &\SWR(\FfcCE_\mathrm{PE}, \FfcCE_\mathrm{general}) = 1.
    \end{aligned}
\end{equation*}
% \end{theorem}
\end{restatable}

\cref{thm:NEG:ratio:CCE} tells that, while the permutation-equivariant (C)CE approximator class may not be able to find all the (C)CE in a game, it can keep the social welfare of the output solutions.

However, when considering equivariant NE approximators, we have the following negative result:
\begin{restatable}{theorem}{thmNEGRatioNE}
% \begin{theorem}
\label{thm:NEG:ratio:NE}
Given $N,M \ge 2$, let ${\Ff}^\NE_\mathrm{OPI},{\Ff}^\NE_\mathrm{PPE}$ and ${\Ff}^\NE_\mathrm{both}$ be the function classes (target on games with number of players $n \le N$ and $|\Aa_i|\le M$) of all the NE approximators with OPI, PPE and both.
Denote the function class of all the NE approximators as $\FfNE_\mathrm{general}$.
Then we have
\begin{align}
    \SWR({\Ff}^\NE_\mathrm{OPI}, \FfNE_\mathrm{general}) &=  \frac{1}{M^{N-1}},\label{eq:NEG:ratio:OPI}\\
    \SWR({\Ff}^\NE_\mathrm{PPE}, \FfNE_\mathrm{general}) &\le \frac{1}{M},\label{eq:NEG:ratio:PPE}\\
    \SWR({\Ff}^\NE_\mathrm{both}, \FfNE_\mathrm{general}) &= \frac{1}{M^{N-1}}.\label{eq:NEG:ratio:both}
\end{align}
Additionally, when $M \ge 3$, denote by $\widetilde\Ff^\NE_\mathrm{both}$ the function class of all the NE oracles (functions that always output exact NE solutions of the input games) with both PPE and OPI, and by  $\widetilde\Ff^\NE_\mathrm{general}$ the function class of all the NE oracles.
Then we have
\begin{equation}\label{eq:NEG:ratio:both:oracle}
    \SWR(\widetilde{\Ff}^\NE_\mathrm{both}, \widetilde\Ff^\NE_\mathrm{general}) = 0.
\end{equation}
% \end{theorem}
\end{restatable}

The proof is done by construction (See \cref{prf:thm:NEG:ratio:NE} for details). 
As an illustration of \cref{eq:NEG:ratio:both:oracle}, consider a bimatrix game with the following payoff:
\begin{equation*}
    u = \begin{bmatrix}
    {1,1} & 0,0 & 0,\frac{1}{2}+\varepsilon
    \\
    0,0 & {1,1} & 0,\frac{1}{2}+\varepsilon
    \\
    \frac{1}{2}+\varepsilon,0 & \frac{1}{2}+\varepsilon,0 & \varepsilon,\varepsilon
    \end{bmatrix}
\end{equation*}
for $\epsilon \in (0, \frac{1}{2})$.
The maximum NE (the upper-left corner of $u$) social welfare is $2$, which can be found by at least one NE oracle in $ \widetilde\Ff^\NE_\mathrm{general}$.
However, the only NE (the lower-right corner of $u$) that the NE oracles in $\widetilde{\Ff}^\NE_\mathrm{both}$ could find only has a social welfare of $2\epsilon$.
As a result, 
\begin{equation*}
    \mathrm{SWR}_{2,3}(\widetilde{\Ff}^\NE_\mathrm{both}, \widetilde\Ff^\NE_\mathrm{general}) \le \frac{2\epsilon}{2} = \epsilon,
\end{equation*}
which goes to zero as $\epsilon \to 0$.
Recall that we always have $\mathrm{SWR}_{N,M} \ge 0$, thus \cref{eq:NEG:ratio:both:oracle} holds when $N=2$ and $M=3$.

\cref{thm:NEG:ratio:NE} tells that equivariant NE approximators may lose some social welfare while enjoying better generalizability.
Such a result inspires us to balance generalizability and social welfare when designing the NE approximator architecture.

% \begin{theorem}
% \label{thm:NEG:2}

% If we withdraw the condition in Theorem \ref{thm:NE:E} that $\fNE$ satisfies OPI, then there exists some $\fNE\in\Ff$ and a payoff matrice distribution $\Dd$ such that
% \begin{equation*}
%     \EE_{u \sim \Dd}[\Ee((\Pp \fNE)(u), u)] > \EE_{u \sim \Dd}[\Ee(\fNE(u), u)]
% \end{equation*}

% \end{theorem}

% This theorem means that orbit average might cause a larger expected exploibility.

% -----

% \section{Experiments(maybe)}

\section{Conclusion and Future Work}
\label{sec:conclusion}

In this paper, we theoretically analyze the benefits and limitations of equivariant equilibrium approximators, including player-permutation-equivariant (PPE) and opponent-permutation-invariant (OPI) NE approximator, and permutation-equivariant (PE) CE and CCE approximators.
For the benefits, we first show that these equivariant approximators enjoy better generalizability. 
To get the result, we derive the generalization bounds and sample complexities based on covering numbers, and then we prove that the symmetric approximators have lower covering numbers.
We then show that the equivariant approximators can decrease the exploitability when the payoff distribution is invariant under permutation.
For the limitations, we find the equivariant approximators may fail to find some equilibria points due to their limited representativeness caused by equivariance.
Besides, while equivariant (C)CE approximators can keep the social welfare, the equivariant NE approximators reach a small worst social welfare ratio comparing to the general approximators.
Such a result indicates that equivariance may reduce social welfare; therefore, we'd better balance the generalizability and social welfare when we design the architectures of NE approximators.

% \paragraph{Future work}
As for future work, since in our paper we assume the training and testing payoff distribution are the same, an interesting topic is to study the benefits of equivariant approximators under the payoff distribution shift.
Moreover, since we consider fixed and discrete action space, another interesting future direction is to analyze the benefits of equivariant approximators in varying or continuous action space.

% \section*{Acknowledgements}
% \todo[inline]{}
	
	\bibliographystyle{plainnat}
	\bibliography{reference}
	
	\newpage
\appendix
\onecolumn
\section{Omitted Proofs in \cref{sec:symmetric}}

\subsection{Useful Lemma}
We first introduce a lemma, which will be frequently used in the following proofs.
\begin{lemma}
\label{lem:u}
$\forall i, j \in [n], \rho_i\in\Gg_i$ we have
$(\rho_i u)_j(\sigma_i,\sigma_{-i}) = u_j (\rho_i^{-1}\sigma_i,\sigma_{-i})$
and 
$
    (\rho_i u)_j(\pi) = u_j(\rho_i^{-1} \pi)
$
\end{lemma}

\begin{proof}
    Define $\widehat{a}_i \coloneqq \rho_i^{-1} a_i$. For product strategy $\sigma=(\sigma_i)_{i\in[n]}$,
    \begin{align*}
        (\rho_i u)_j(\sigma_i,\sigma_{-i})
        =& \sum_{a_i\in\Aa_i} \sum_{a_{-i}\in\Aa_{-i}} (\rho_i u)_j(a_i, a_{-i}) \cdot \sigma_i (a_i) \cdot \sigma_{-i}(a_{-i}) &
        \\
        =& \sum_{a_i\in\Aa_i} \sum_{a_{-i}\in\Aa_{-i}} u_j(\rho_i^{-1} a_i, a_{-i}) \cdot \sigma_i (a_i) \cdot \sigma_{-i}(a_{-i}) &
        \\
        =& \sum_{a_i\in\Aa_i} \sum_{a_{-i}\in\Aa_{-i}} u_j(\rho_i^{-1} a_i, a_{-i}) \cdot (\rho_i^{-1}\sigma_i) (\rho_i^{-1} a_i) \cdot \sigma_{-i}(a_{-i}) &
        \\
        =& \sum_{\widehat{a}_i\in\Aa_i} \sum_{a_{-i}\in\Aa_{-i}} u_j(\widehat{a}_i, a_{-i}) \cdot (\rho_i^{-1}\sigma_i) (\widehat{a}_i) \cdot \sigma_{-i}(a_{-i}) &
        \\
        =& u_j(\rho_i^{-1}\sigma_i,\sigma_{-i}) &
    \end{align*}
    
    For joint strategy $\pi$, 
    \begin{align*}
        (\rho_i u)_j(\pi)
        =& \sum_{a_i \in \Aa_i} \sum_{a_{-i}\in\Aa_{-i}} (\rho_i u_j)(a_i,a_{-i}) \cdot \pi(a_i,a_{-i}) 
        \\
        =& \sum_{a_i \in \Aa_i} \sum_{a_{-i}\in\Aa_{-i}} u_j(\rho_i^{-1} a_i,a_{-i}) \cdot \pi(a_i,a_{-i})&
        \\
        =& \sum_{a_i \in \Aa_i} \sum_{a_{-i}\in\Aa_{-i}} u_j(\rho_i^{-1} a_i,a_{-i}) \cdot (\rho_i^{-1} \pi)(\rho_i^{-1} a_i,a_{-i})&
        \\
        =& \sum_{\widehat{a}_i \in \Aa_i} \sum_{a_{-i}\in\Aa_{-i}} u_j(\widehat{a}_i,a_{-i}) \cdot (\rho_i^{-1} \pi)(\widehat{a}_i,a_{-i})&
        \\
        =& u_j(\rho_i^{-1}\pi)
    \end{align*}
\end{proof}

\subsection{Proof of \cref{lem:perm:NE}}\label{prf:lem:perm:NE}
\lemPermNE*

\begin{proof}
For player $i$, we have
\begin{align*}
    \Ee_i(\rho_i \sigma,\rho_i u)
    =& \max_{a_i \in \Aa_i} \rho_i u_i(a_i,\rho_i\sigma_{-i}) - \rho_i u_i(\rho_i\sigma)
    = \max_{a_i \in \Aa_i} \rho_i u_i(a_i,\sigma_{-i}) - \rho_i u_i(\rho_i\sigma_i,\sigma_{-i})
    \\
    =& \max_{a_i \in \Aa_i} u_i(\rho_i^{-1}a_i ,\sigma_{-i}) - u_i(\rho_i^{-1}\rho_i\sigma_i,\sigma_{-i})
    \overset{(a)}{=} \max_{a_i\in\Aa_i} u_i(a_i,\sigma_{-i}) - u_i(\sigma_i,\sigma_{-i}) 
    = \Ee_i(\sigma,u),
\end{align*}
where  $(a)$ holds since $\rho_i$ is a bijection on $\Aa_i$.
For player $j \neq i$, we have
\begin{align*}
    \Ee_j(\rho_i\sigma,\rho_i u)
    =& \max_{a_j\in\Aa} \rho_i u_j(a_j,\rho_i\sigma_{-j}) - \rho_i u_j(\rho_i\sigma)
    = \max_{a_j\in\Aa_j} u_j(a_j,\rho_i^{-1}\rho_i \sigma_{-j}) - u_j(\rho_i^{-1}\rho_i \sigma)
    \\
    =& \max_{a_j\in\Aa_j} u_j(a_j,\sigma_{-j}) - u_j(\sigma)
    = \Ee_j(\sigma,u)
\end{align*}

From above, we have $\Ee(\rho_i\sigma,\rho_i u) = \Ee(\sigma,u)$, thus if $\sigma$ is a $\varepsilon$-NE of $\Gamma_u$, then $\rho_i\sigma$ must be a $\varepsilon$-NE of $\Gamma_{\rho_i u}$.
\end{proof}
% -----

\subsection{Proof of \cref{lem:perm:CCE}}
\label{prf:lem:perm:CCE}

\lemPermCCE*

\paragraph{CCE}
For player $i$, we have
\begin{align*}
    \Ee_i(\rho_i \pi,\rho_i u)
    =& \max_{a_i \in \Aa_i} (\rho_i u_i)(a_i,(\rho_i \pi)_{-i}) - (\rho_i u_i)(\rho_i \pi_i)
    \\
    =& \max_{a_i \in \Aa_i} (\rho_i u_i)(a_i,(\rho_i \pi)_{-i}) - u_i(\rho_i^{-1} \rho_i \pi_i)
    \\
    =& \max_{a_i \in \Aa_i} (\rho_i u_i)(a_i,(\rho_i \pi)_{-i}) - u_i(\pi_i)
    \\
    =& \max_{a_i \in \Aa_i} \sum_{b\in\Aa} (\rho_i u_i)(a_i,b_{-i})\cdot (\rho_i\pi)(b) - u_i(\pi_i)
    \\
    =& \max_{a_i \in \Aa_i} \sum_{b_i\in\Aa_i,b_{-i}\in\Aa_{-i}} u_i(\rho_i^{-1} a_i,b_{-i})\cdot \pi(\rho_i^{-1} b_i, b_{-i}) - u_i(\pi_i)
    \\
    =& \max_{a_i \in \Aa_i} \sum_{b_i\in\Aa_i,b_{-i}\in\Aa_{-i}} u_i(a_i,b_{-i})\cdot \pi(b_i, b_{-i}) - u_i(\pi_i) & &,\text{$\rho_i$ is a bijection on $\Aa_i$}
    \\
    =& \Ee_i(\pi,u)
\end{align*}

For player $j \neq i$, we have
\begin{align*}
    \Ee_j(\rho_i \pi,\rho_i u)
    =& \max_{a_j\in\Aa_j} (\rho_i u_j)(a_j,(\rho_i\pi)_{-j}) - (\rho_i u_j)(\rho_i \pi_j)
    \\
    =& \max_{a_j\in\Aa_j} (\rho_i u_j)(a_j,(\rho_i\pi)_{-j}) - u_j(\rho_i^{-1}\rho_i \pi_j)
    \\
    =& \max_{a_j\in\Aa_j} (\rho_i u_j)(a_j,(\rho_i\pi)_{-j}) - u_j(\pi_j)
    \\
    =& \max_{a_j\in\Aa_j} \sum_{b\in\Aa} (\rho_i u_j)(a_j,b_{-j})\cdot(\rho_i \pi)(b) - u_j(\pi_j)
    \\
    =& \max_{a_j\in\Aa_j} \sum_{b_i\in\Aa_i,b_{-i}\in\Aa_{-i}} u_j(a_j,(b_{-j})_{-i},\rho_i^{-1} b_i)\cdot\pi(\rho_i^{-1} b_i,b_{-i}) - u_j(\pi_j)
    \\
    =& \max_{a_j\in\Aa_j} \sum_{b_i\in\Aa_i,b_{-i}\in\Aa_{-i}} u_j(a_j,(b_{-j})_{-i},b_i)\cdot\pi(b_i,b_{-i}) - u_j(\pi_j) & &,\text{$\rho_i$ is a bijection on $\Aa_i$}
    \\
    =& \Ee_j(\pi,u)
\end{align*}
Thus, we have $\Ee(\rho_i\pi,\rho_i u) = \Ee(\pi,u)$. 
Thus, if $\pi$ is a $\varepsilon$-CCE of $\Gamma_u$, then $\rho_i\pi$ must be a $\varepsilon$-CCE of $\Gamma_{\rho_i u}$.

\paragraph{CE}
For player $j\ne i$, we have
\begin{align*}
    \Ee_j^\CE(\rho_i \pi,\rho_i u)
    =& \max_{\phi_j:\Aa_j\to\Aa_j} \sum_{a\in\Aa} (\rho_i \pi)(a)\cdot (\rho_i u_j)(\phi_j(a_j),a_{-j}) - (\rho_i u_j)(\rho_i \pi)
    \\
    =& \max_{\phi_j:\Aa_j\to\Aa_j} \sum_{a\in\Aa} \pi(\rho_i^{-1}a_i, a_{-i})\cdot u_j(\phi_j(a_j),a_{-i,j},\rho_i^{-1} a_i) - u_j(\pi)
    \\
    =& \max_{\phi_j:\Aa_j\to\Aa_j} \sum_{a\in\Aa} \pi(a_i, a_{-i}) \cdot u_j(\phi_j(a_j),a_{-i,j}, a_i) - u_j(\pi) & &,\text{$\rho_i$ is a bijection on $\Aa_i$}
    \\
    =& \Ee_j^\CE(\pi,u)
\end{align*}

For player $i$, we define operator $\Bar{\rho}_i$ as $(\bar{\rho}_i \phi_i)(a_i) = \rho_i^{-1} \phi_i(\rho_i a_i)$.
We can verify that $\Bar{\rho}_i$ is a bijection on $\{ \phi_i:\Aa_i\to\Aa_i\}$, because $\Bar{\cdot}$ is a homomorphism in the sense that
$\overline{\rho_i^{1}} \circ \overline{\rho_i^{2}} = \overline{\rho_i^{2}\rho_i^{1}}$
and $\Bar{\cdot}$ maps the identity mapping of $\Aa_i$ to the identity mapping of $\{\Aa_i \to \Aa_i\}$.
Specifically, 
% $
%     \overline{\rho_i^{1}} \circ \overline{\rho_i^{2}} \phi_i(a_i)
%     = (\rho_i^{1})^{-1} (\overline{\rho_i^{2}} \phi_i)(\rho_i^{1} a_i)
%     = (\rho_i^{1})^{-1} (\rho_i^{2})^{-1} \phi_i(\rho_i^{2} \rho_i^{1} a_i)
%     = \overline{\rho_i^{2}\rho_i^{1}} \phi_i(a_i)
% $ and 
% $
%     \overline{e_i} \phi_i(a_i)
%     = e_i^{-1} \phi_i(e_i a_i)
%     = \phi_i(a_i)
% $,
\begin{align*}
    \overline{\rho_i^{1}} \circ \overline{\rho_i^{2}} \phi_i(a_i)
    = (\rho_i^{1})^{-1} (\overline{\rho_i^{2}} \phi_i)(\rho_i^{1} a_i)
    = (\rho_i^{1})^{-1} (\rho_i^{2})^{-1} \phi_i(\rho_i^{2} \rho_i^{1} a_i)
    = \overline{\rho_i^{2}\rho_i^{1}} \phi_i(a_i),
\end{align*}
and
\begin{align*}
    \overline{e_i} \phi_i(a_i)
    = e_i^{-1} \phi_i(e_i a_i)
    = \phi_i(a_i).
\end{align*}
% where $\overline{\rho_i^{-1}}$ is the inverse element of $\Bar{\rho_i}$ and $\Bar{e_i}$ is the identity mapping, so is a bijection. Thus we must have $\Bar{\rho_i}$ is a bijection on $\{\phi_i:\Aa_i \to \Aa_i\}$.

Based on $\bar\rho_i$, we have
\begin{align*}
    & \Ee_i^\CE(\rho_i \pi,\rho_i u)
    \\
    =& \max_{\phi_i:\Aa_i\to\Aa_i} \sum_{a\in\Aa} (\rho_i \pi)(a)\cdot (\rho_i u_i)(\phi_i(a_i),a_{-i}) - u_i(\pi)
    \\
    =& \max_{\phi_i:\Aa_i\to\Aa_i} \sum_{a\in\Aa} \pi(\rho_i^{-1} a_i, a_{-i}) u_i(\rho_i^{-1} \phi_i(a_i),a_{-i}) - u_i(\pi)
    \\
    =& \max_{\phi_i:\Aa_i\to\Aa_i} \sum_{a\in\Aa} \pi(\rho_i^{-1} a_i, a_{-i}) u_i(\rho_i^{-1} \phi_i(\rho_i (\rho_i^{-1} a_i)),a_{-i}) - u_i(\pi)
    \\
    =& \max_{\phi_i:\Aa_i\to\Aa_i} \sum_{a\in\Aa} \pi(a_i, a_{-i}) u_i(\rho_i^{-1} \phi_i(\rho_i a_i),a_{-i}) - u_i(\pi)& &,\text{$\rho_i$ is a bijection on $\Aa_i$}
    \\
    =& \max_{\phi_i:\Aa_i\to\Aa_i} \sum_{a\in\Aa} \pi(a_i, a_{-i}) u_i((\Bar{\rho}_i \phi_i)( a_i),a_{-i}) - u_i(\pi)
    \\
    =& \max_{\phi_i:\Aa_i\to\Aa_i} \sum_{a\in\Aa} \pi(a_i, a_{-i}) u_i(\phi_i( a_i),a_{-i}) - u_i(\pi)& &,\text{$\Bar{\rho}_i$ is a bijection on $\{\Aa_i\to\Aa_i\}$}
    \\
    =& \Ee_i^\CE(\pi,u)
\end{align*}
Thus, we have $\Ee(\rho_i\pi,\rho_i u) = \Ee(\pi,u)$, thus if $\pi$ is a $\varepsilon$-CE of $\Gamma_u$, then $\rho_i\pi$ must be a $\varepsilon$-CE of $\Gamma_{\rho_i u}$.

% -----

\subsection{Proof of \cref{lem:NE:OiPi} to \cref{lem:CCE:Q}}

\lemNEOiPi*

\begin{proof}
\label{prf:lem:NE:OiPi}
$\forall j\ne i, \rho_0\in\Gg_i$, for operator $\Oo_i$ we have
\begin{align*}
    (\Oo_i \fNE)(\rho_0 u)_j =& \frac{1}{|\Aa_i|!} \sum_{\rho_i\in\Gg_i} \fNE (\rho_i \rho_0 u)_j 
    \overset{(a)}{=} \frac{1}{|\Aa_i|!} \sum_{\widehat{\rho}_i\in\Gg_i} \fNE (\widehat{\rho}_i u)_j
    = (\Oo_i \fNE)(u)_j
\end{align*}
where in $(a)$ we define $\widehat{\rho}_i=\rho_i \rho_0$, and $(a)$  holds since $\rho_0$ is a bijection on $\Gg_i$.
As a result, $\Oo_i \fNE$ is $i$-PI.

For operator $\Pp_i$ we have
\begin{align*}
    (\Pp_i \fNE)(\rho_0 u)_i =& \frac{1}{|\Aa_i|!} \sum_{\rho_i\in\Gg_i} \rho_i^{-1}\fNE (\rho_i \rho_0 u)_j
    = \rho_0 \frac{1}{|\Aa_i|!} \sum_{\rho_i\in\Gg_i} \rho_0^{-1} \rho_i^{-1}\fNE (\rho_i \rho_0 u)_j
    \\
    =& \rho_0 \frac{1}{|\Aa_i|!} \sum_{\widehat{\rho}_i\in\Gg_i} \widehat{\rho}_i^{-1}\fNE (\widehat{\rho}_i u)_j
    =\rho_0 (\Pp_i \fNE)(u)_i,
\end{align*}
therefore $\Pp_i \fNE$ is $i$-PE.

% Notice $\Qq_i = \Oo_i \circ \Pp_i$, thus $\Qq_i \fNE$ satisfied $i$-PI and $i$-PE.

If $\fNE$ is already $i$-PI, $\forall j \neq i$ we have
\begin{align*}
    \Oo_i \fNE(u)_j =& \frac{1}{|\Aa_i|!} \sum_{\rho_i\in\Gg_i}  \fNE(\rho_i u)_j
    = \frac{1}{|\Aa_i|!} \sum_{\rho_i\in\Gg_i}  \fNE(u)_j
    = \fNE(u)_j,
\end{align*}
and $\Oo_i \fNE(u)_i = \fNE(u)_i$ according to definition of $\Oo_i$. Therefore, $\Oo_i \fNE = \fNE$ for $i$-PI $\fNE$.

If $\fNE$ is already $i$-PE, we have
\begin{align*}
    \Pp_i \fNE(u)_i =& \frac{1}{|\Aa_i|!} \sum_{\rho_i\in\Gg_i} \rho_i^{-1} \fNE(\rho_i u)_i
    = \frac{1}{|\Aa_i|!} \sum_{\rho_i\in\Gg_i} \rho_i^{-1}\rho_i \fNE(u)_i
    = \frac{1}{|\Aa_i|!} \sum_{\rho_i\in\Gg_i} \fNE(u)_i
    =\fNE(u)_i,
\end{align*}
and $\forall j \neq i, \Pp_i \fNE(u)_j = \fNE(u)_j$ according to definition of $\Pp_i$. Therefore, $\Pp_i \fNE = \fNE$ for $i$-PE $\fNE$.

\end{proof}

\lemNEOP*
\begin{proof}
A direct inference from \cref{lem:NE:OiPi}
\end{proof}

\lemCCEQ*
\begin{proof}
\label{prf:lem:CCE:Q}
$\forall \rho_0\in\Gg_i$, we have

\begin{align*}
    (\Qq_i \fcCE)(\rho_0 u) =& \frac{1}{|\Aa_i|!} \sum_{\rho_i \in \Gg_i} \rho_i^{-1} \fcCE (\rho_i \rho_0 u)
    = \rho_0 \frac{1}{|\Aa_i|!} \sum_{\rho_i \in \Gg_i} \rho_0^{-1} \rho_i^{-1} \fcCE (\rho_i \rho_0 u)
    \\
    =& \rho_0 \frac{1}{|\Aa_i|!} \sum_{\widehat{\rho}_i \in \Gg_i} \widehat{\rho}_i^{-1} \fcCE (\widehat{\rho}_i u)
    = \rho_0 (\Qq_i \fcCE)(u)
\end{align*}

If $f^\text{(C)CE}$ is already $i$-PE, we have
\begin{align*}
    \Qq_i f^\text{(C)CE}(u) 
    =& \frac{1}{|\Aa_i|!}\sum_{\rho_i\in\Gg_i} \rho_i^{-1} f^\text{(C)CE}(\rho_i u)
    = \frac{1}{|\Aa_i|!}\sum_{\rho_i\in\Gg_i} \rho_i^{-1} \rho_i f^\text{(C)CE}(u)
    = \frac{1}{|\Aa_i|!}\sum_{\rho_i\in\Gg_i} f^\text{(C)CE}(u)
    = f^\text{(C)CE}(u)
\end{align*}

\end{proof}

% -----
\subsection{Proof of \cref{lem:subset}}
\label{prf:lem:subset}
\lemSubset*

\begin{proof}
We prove the three claims below.
\begin{enumerate}
    \item $\Xx \Ff_\Xx \subseteq \Ff_\Xx$.

    \item $\Xx^2 \Ff_\Xx = \Xx \Ff_\Xx$.

    \item If $\Xx \Yy = \Yy \subseteq \Ff_\Xx$, then $\Yy \subseteq \Xx \Ff_\Xx$
\end{enumerate}

The first claim holds because $\Ff_\Xx$ is closed under $\Xx$, and the second claim holds because $\Xx$ is idempotent. For the third claim, from $\Yy \subseteq \Ff_\Xx$ we know $\Xx\Yy \subseteq \Xx \Ff_\Xx$, then $\Yy = \Xx\Yy \subseteq \Xx\Ff_\Xx$.

We immediately know $\Xx\Ff_\Xx$ is the largest subset
of $\Ff_\Xx$ that is invariant under $\Xx$.
\end{proof}

\section{Omitted Proofs in \cref{sec:benefits}}

\subsection{Proof of \cref{thm:GB}}\label{prf:thm:GB}
\thmGB*

Some of the proof techniques come from \citet{dutting2019optimal} and \citet{duan2021towards}.
We first introduce some useful lemmas.
Denote $\ell: \Ff \times \Uu \to \RR$ as the loss function (such as $\ell(f, u) \coloneqq \Ee(f(u), u)$). 
We measure the capacity of the composite function class $\ell \circ \Ff$ using the empirical Rademacher complexity~\citep{bartlett2002rademacher} on the training set $S$, which is defined as:
\begin{equation*}
\begin{aligned}
	\Rr_S(\ell \circ \Ff) \coloneqq
	\frac{1}{m}\EE_{\bm x \sim \{+1,-1\}^m}\Big[\sup_{f\in\Ff} \sum_{i=1}^m x_i \cdot \ell(f, u^{(i)}) \Big],
\end{aligned}
\end{equation*}
where $\bm x$ is distributed i.i.d. according to uniform distribution in $\{+1,-1\}$.
We have

\begin{lemma}[\citet{shalev2014understanding}]
\label{lem:GB:Rad}
    Let $S$ be a training set of size $m$ drawn i.i.d. from distribution $\Dd$ over $\Uu$.
    Then with probability at least $1 - \delta$ over draw of $S$ from $\Dd$, for all $f \in \Ff$,
    \begin{equation*}
    	\EE_{u\sim\Dd}[\ell(f, u)] - \frac{1}{m}\sum_{u\in S}\ell(l, u)  \le 2\Rr_S(\ell \circ \Ff) + 4\sqrt{\frac{2\ln(4/\delta)}{m}}
    \end{equation*}
\end{lemma}

\begin{lemma}
\label{lem:GB}
    If $|\ell(\cdot)| \le c$ for constant $c > 0$ and $\forall f, f' \in \Ff, | \ell(f, u) - \ell(f', u)| \le L\norm{f - f'}_{\infty}$, then we have
    \begin{equation*}
    	\EE_{u\sim\Dd}[\ell(f, u)] - \frac{1}{m}\sum_{u\in S}\ell(l, u)  \le 2\inf_{r>0}\left\{ c\sqrt{\frac{2\ln\Nn_{\infty}(\Ff, r)}{m}}+ Lr \right\} + 4\sqrt{\frac{2\ln(4/\delta)}{m}}
    \end{equation*}
\end{lemma}
\begin{proof}
    For function class $\Ff$, let $\Ff_r$ with $|\Ff_r| = \Nn_{\infty}(\Ff, r)$ be the function class that $r$-covers $\Ff$ for some $r>0$. 
    Similarly, $\forall f \in \Ff$, denote $f_r \in \Ff_r$ be the function that $r$-covers $f$. 
    We have
    \begin{equation}
    \label{eq:Rad}
	\begin{aligned}
		\Rr_S(\ell \circ \Ff) 
		=& \frac{1}{m}\EE_{\bm x}\Big[\sup_{f\in\Ff} \sum_{i=1}^m x_i \cdot \ell(f, u^{(i)}) \Big]
		\\
		=& \frac{1}{m}\EE_{\bm x}\Big[\sup_{f\in\Ff} \sum_{i=1}^m x_i \cdot \big( \ell(f_r, u^{(i)})
		+  \ell(f, u^{(i)})  -  \ell(f_r, u^{(i)})\big)  \Big]
		\\
		\le& \frac{1}{m}\EE_{\bm x}\Big[\sup_{f_r \in \Ff_r} \sum_{i=1}^m x_i \cdot  \ell(f_r, u^{(i)})  \Big]
		+ \frac{1}{m}\EE_{\bm x}\Big[\sup_{f\in \Ff} \sum_{i=1}^m |x_i \cdot Lr| \Big] & & ,|\ell(f, u) - \ell(f_r, u)| \le L\norm{f - f_r}_{\infty} = Lr
		\\
		{\le}& \sup_{f_r\in \Ff_r}\sqrt{\sum_{i=1}^m \ell^2(f_r, u^{(i)})} \cdot \frac{\sqrt{2\ln\Nn_{\infty}(\Ff, r)}}{m} + \frac{Lr}{m}\EE_{\bm x}\norm{\bm x} & &,\text{the first term holds by Massart's lemma}
		\\
		\le& \sqrt{c^2m} \cdot \frac{\sqrt{2\ln\Nn_{\infty}(\Ff, r)}}{m} + \frac{Lr}{m}\EE_{\bm x}\norm{\bm x}
		\\
		\le& c\sqrt{\frac{2\ln\Nn_{\infty}(\Ff, r)}{m}}+ Lr,
	\end{aligned}
    \end{equation}
    % where the second term of $(a)$ holds from $|\ell(f, u) - \ell(f_r, u)| \le L\norm{f - f_r}_{\infty} = Lr$, and the first term of $(b)$ holds by Massart's lemma~\citep{shalev2014understanding}.
    
    Combining \cref{lem:GB:Rad} and \cref{eq:Rad}, we get
    
    \begin{equation*}
    	\EE_{u\sim\Dd}[\ell(f, u)] - \frac{1}{m}\sum_{u\in S}\ell(l, u)  \le 2\inf_{r>0}\left\{ c\sqrt{\frac{2\ln\Nn_{\infty}(\Ff, r)}{m}}+ Lr \right\} + 4\sqrt{\frac{2\ln(4/\delta)}{m}}
    \end{equation*}
\end{proof}

\subsubsection{NE Approximator}
\begin{lemma}
	\label{lem:NE:L_sigma}
	For arbitrary product mixed strategy $\sigma$ and $\sigma'$, we have 
	\begin{equation*}
		\begin{aligned}
			|\Ee(\sigma,u) - \Ee(\sigma', u)|
			\le 2n\norm{\sigma - \sigma'},
		\end{aligned}
	\end{equation*}
\end{lemma}
\begin{proof}
    $\forall \sigma, \sigma'$, we define $y_{-j} \coloneqq (\sigma_1, \dots, \sigma_{j-1}, \sigma'_{j+1}, \dots, \sigma'_n)$. 
    Then, $\forall i \in [n]$ we have
    \begin{equation*}
        \begin{aligned}
            |u_i(\sigma) - u_i(\sigma')|
            =& |u_i(\sigma_1, \sigma_2, \dots, \sigma_n) - u_i(\sigma', \sigma'_2, \dots, \sigma'_n)| 
            \\
            =& \Big| \sum_{j=1}^n \Big(u_i(\sigma_1, \dots, \sigma_j, \sigma'_{j+1}, \dots, \sigma'_n) 
            - u_i(\sigma_1, \dots, \sigma'_j, \sigma'_{j+1}, \dots, \sigma'_n) \Big)\Big| 
            \\
            =& \Big| \sum_{j=1}^n \Big(u_i(\sigma_j, y_{-j}) - u_i(\sigma'_j, y_{-j}) \Big)\Big|
            \\
            =& \Big| \sum_{j=1}^n \sum_{a_j}(\sigma_j(a_j) - \sigma'_j(a_j)) \sum_{a_{-j}} u_i(a_j, a_{-j})y_{-j}(a_{-j})  \Big|
            \\
            \le& \sum_{j=1}^n \sum_{a_j}\Big|\sigma_j(a_j) - \sigma'_j(a_j)\Big|\sum_{a_{-j}} u_i(a_j, a_{-j})y_{-j}(a_{-j}) 
            \\
            {\le}& \sum_{j=1}^n \sum_{a_j}\Big|\sigma_j(a_j) - \sigma'_j(a_j)\Big|\sum_{a_{-j}}y_{-j}(a_{-j}) & &,u_i(\cdot) \in [0, 1]
            \\
            \le& \sum_{j=1}^n \sum_{a_j\in A_j}\Big|\sigma_j(a_j) - \sigma'_j(a_j)\Big|
            \le n\max_{j \in [n]} \sum_{a_j\in A_j}\Big|\sigma_j(a_j) - \sigma'_j(a_j)\Big|
            \\
            =& n\norm{\sigma - \sigma'},
        \end{aligned}
    \end{equation*}
    % where $(a)$ holds since $u_i(\cdot) \in [0, 1]$.
    Therefore, $\forall a_i \in A_i$, 
    \begin{equation*}
        \begin{aligned}
            u_i(a_i, \sigma_{-i}) - u_i(\sigma) 
            =& u_i(a_i, \sigma_{-i}) - u_i(a_i, \sigma'_{-i}) + u_i(a_i, \sigma'_{-i}) - u_i(\sigma') + u_i(\sigma') - u_i(\sigma)
            \\
            \le& n\norm{\sigma - \sigma'} + \Ee(\sigma', u) + n\norm{\sigma - \sigma'}
            \\
            =& \Ee(\sigma', u) + 2n\norm{\sigma - \sigma'}.
        \end{aligned}
    \end{equation*}
    Based on that, we get
    \begin{equation*}
        \begin{aligned}
            \Ee(\sigma, u) 
            = \max_{i\in N, a_i \in A_i}[u_i(a_i, \sigma_{-i}) - u_i(\sigma)]
            \le \Ee(\sigma', u) + 2n\norm{\sigma - \sigma'}
        \end{aligned}
    \end{equation*}
    Similarly, we also have
    \begin{equation*}
        \Ee(\sigma', u) 
        \le \Ee(\sigma, u) + 2n\norm{\sigma - \sigma'}
    \end{equation*}
\end{proof}

Based on \cref{lem:NE:L_sigma}, $\forall f, f' \in \FfNE$, we have
\begin{equation*}
	\Ee(f(u), u) - \Ee(f'(u), u) \le 2\norm{f(u) - f'(u)} \le 2\norm{f - f'}_{\infty}
\end{equation*}
Considering that $|\Ee(\cdot)| \le 1$, according to \cref{lem:GB}, we have:
\begin{equation*}
	\EE_{u\sim\Dd}[\Ee(\fNE(u), u)] - \frac{1}{m}\sum_{u\in S}\Ee(\fNE(u), u) \le 2\cdot \inf_{r>0}\Big\{ \sqrt{\frac{2\ln\Nn_{\infty}(\FfNE, r)}{m}}+ 2nr \Big\}
	+ 4\sqrt{\frac{2\ln(4/\delta)}{m}}
\end{equation*}

\subsubsection{CCE Approximator}
\begin{lemma}
	\label{lem:CCE:L_pi}
	For arbitrary joint mixed strategy $\pi$ and $\pi'$, we have 
	\begin{equation*}
		|\Ee(\pi,u) - \Ee(\pi', u)|
		\le 2\norm{\pi - \pi'},
	\end{equation*}
\end{lemma}
\begin{proof}
	$\forall \pi, \pi', \forall i \in [n]$ we have
	\begin{equation}
	\label{eq:u:L_pi}
	\begin{aligned}
		|u_i(\pi) - u_i(\pi')| = \sum_{a\in\Aa}(\pi(a) - \pi'(a))u_i(a)
		\overset{(a)}{\le} \sum_{a\in\Aa}|\pi(a) - \pi'(a)|
		= \norm{\pi - \pi'}
	\end{aligned}
	\end{equation}
	where $(a)$ holds since $u_i(\cdot) \in [0, 1]$.
	Therefore, $\forall a_i \in A_i$, 
	\begin{equation*}
		\begin{aligned}
			u_i(a_i, \pi_{-i}) - u_i(\pi) 
			=& u_i(a_i, \pi_{-i}) - u_i(a_i, \pi'_{-i}) + u_i(a_i, \pi'_{-i}) - u_i(\pi') + u_i(\pi') - u_i(\pi)
			\\
			\le& \norm{\pi - \pi'} + \Ee(\pi', u) + \norm{\pi - \pi'}
			\\
			=& \Ee(\pi', u) + 2\norm{\pi - \pi'}.
		\end{aligned}
	\end{equation*}
	Based on that, we get
	\begin{equation*}
		\begin{aligned}
			\Ee(\pi, u) 
			= \max_{i\in N, a_i \in A_i}[u_i(a_i, \pi_{-i}) - u_i(\pi)]
			\le \Ee(\pi', u) + 2\norm{\pi - \pi'}
		\end{aligned}
	\end{equation*}
	Similarly, we also have
	\begin{equation*}
		\Ee(\pi', u) 
		\le \Ee(\pi, u) + 2\norm{\pi - \pi'}
	\end{equation*}
\end{proof}

Based on \cref{lem:CCE:L_pi}, $\forall f, f' \in \FfCCE$, we have
\begin{equation*}
	\Ee(f(u), u) - \Ee(f'(u), u) \le 2\norm{f(u) - f'(u)} \le 2\norm{f - f'}_{\infty}
\end{equation*}
Considering that $|\Ee(\cdot)| \le 1$, according to \cref{lem:GB}, we have:
\begin{equation*}
	\EE_{u\sim\Dd}[\Ee(\fCCE(u), u)] - \frac{1}{m}\sum_{u\in S}\Ee(\fCCE(u), u) \le 2\cdot \inf_{r>0}\Big\{ \sqrt{\frac{2\ln\Nn_{\infty}(\FfCCE, r)}{m}}+ 2r \Big\}
	+ 4\sqrt{\frac{2\ln(4/\delta)}{m}}
\end{equation*}

\subsubsection{CE Approximator}
\begin{lemma}
	\label{lem:CE:L_pi}
	For arbitrary joint mixed strategy $\pi$ and $\pi'$, we have 
	\begin{equation*}
		|\Ee^\CE(\pi,u) - \Ee^\CE(\pi', u)|
		\le 2\norm{\pi - \pi'},
	\end{equation*}
\end{lemma}
\begin{proof}
    $\forall a_i \in A_i, \forall \phi_i$, we have 
	\begin{equation*}
	\begin{aligned}
		\sum_{a\in\Aa}\pi(a)u_i(\phi(a_i), a_{-i}) - u_i(\pi)
		=& \sum_{a\in\Aa}\pi(a)u_i(\phi(a_i), a_{-i}) - \sum_{a\in\Aa}\pi'(a)u_i(\phi(a_i), a_{-i}) \\ &\quad+ \sum_{a\in\Aa}\pi'(a)u_i(\phi(a_i), a_{-i}) - u_i(\pi') + u_i(\pi') - u_i(\pi)
		\\
		\le& \norm{\pi - \pi'} + \Ee^\CE(\pi', u) + \norm{\pi - \pi'}
		\\
		=& \Ee^\CE(\pi', u) + 2\norm{\pi - \pi'}.
	\end{aligned}
	\end{equation*}
	Based on that, we get
	\begin{equation*}
		\begin{aligned}
			\Ee^\CE(\pi, u) 
			= \max_{i\in N}\max_{\phi_i}\sum_{a\in\Aa}\pi(a)u_i(\phi(a_i), a_{-i})  - u_i(\pi)
			\le \Ee^\CE(\pi', u) + 2\norm{\pi - \pi'}
		\end{aligned}
	\end{equation*}
	Similarly, we also have
	\begin{equation*}
		\Ee^\CE(\pi', u) 
		\le \Ee^\CE(\pi, u) + 2\norm{\pi - \pi'}
	\end{equation*}
\end{proof}

Based on \cref{lem:CCE:L_pi}, $\forall f, f' \in \FfCE$, we have
\begin{equation*}
	\Ee^\CE(f(u), u) - \Ee^\CE(f'(u), u) \le 2\norm{f(u) - f'(u)} \le 2\norm{f - f'}_{\infty}
\end{equation*}
Considering that $|\Ee(\cdot)| \le 1$, according to \cref{lem:GB}, we have:
\begin{equation*}
	\EE_{u\sim\Dd}[\Ee^\CE(\fCE(u), u)] - \frac{1}{m}\sum_{u\in S}\Ee^\CE(\fCE(u), u) \le 2\cdot \inf_{r>0}\Big\{ \sqrt{\frac{2\ln\Nn_{\infty}(\FfCE, r)}{m}}+ 2r \Big\}
	+ 4\sqrt{\frac{2\ln(4/\delta)}{m}}
\end{equation*}

\subsection{Proof of \cref{thm:SC}}\label{prf:thm:SC}
\thmSC*
\begin{proof}
For function class $\Ff$ of NE, CE or CCE approximators, according to \cref{lem:NE:L_sigma}, \cref{lem:CCE:L_pi} and \cref{lem:CE:L_pi}, $\forall f, g \in \Ff$ we have
\begin{equation}
\label{eq:L_sigma_r}
    \Ee^{(\CE)}(f(u), u) - \Ee^{(\CE)}(g(u), u) \le L\norm{f(u) - g(u)} \le L\norm{f - g}_{\infty},
\end{equation}
where $L = 2n$ for NE approximators, and $L = 2$ for CE and CCE approximators.

For simplicity, we denote $L_S(f) = \frac{1}{m}\sum_{u\in S}\Ee^{(\CE)}(f(u), u)$ and $L_\Dd(f) = \EE_{u\sim\Dd}[\Ee^{(\CE)}(f(u), u)]$.
let $\Ff_r$ with $|\Ff_r| = \Nn_{\infty}(\Ff, r)$ be the function class that $r$-covers $\Ff$ for some $r>0$. 
$\forall \epsilon \in (0, 1)$, by setting $r = \frac{\epsilon}{3L}$ we have
\begin{equation*}
	\begin{aligned}
		&\PP_{S\sim\Dd^m}\Big[\exists f\in\Ff, \big|L_S(f)
		- L_\Dd(f)\big| > \epsilon \Big] 
		\\
		\le& \PP_{S\sim\Dd^m}\Big[ \exists f\in\Ff, \big|L_S(f) - L_S({f}_r)\big| + \big| L_S({f}_r) - L_\Dd({f}_r)\big|
		+ \big|L_\Dd({f}_r) - L_\Dd(f)\big| 
		> \epsilon \Big] 
		\\
		\overset{(a)}{\le}& \PP_{S\sim\Dd^m}\Big[\exists f\in\Ff, 
		Lr + \big| L_S({f}_r) - L_\Dd({f}_r)\big| + Lr > \epsilon \Big] 
		\\
		\le& \PP_{S\sim\Dd^m}\Big[\exists {f}_r \in {\Ff}_r, \big|L_S({f}_r) - L_\Dd({f}_r) \big| > \epsilon - 2Lr \Big]
		\\
		\overset{(b)}{\le}& \Nn_{\infty}(\Ff, r)\PP_{S\sim\Dd^m}\Big[\big|L_S({f}) - L_\Dd({f}) \big| > \epsilon - 2Lr \Big]
		\\
		\overset{(c)}{\le}& 2\Nn_{\infty}(\Ff, r)\exp(-2m(\epsilon - 2Lr)^2),
		\\
		=& 2\Nn_{\infty}(\Ff, \frac{\epsilon}{3L})\exp(-\frac{2}{9}m\epsilon^2)
	\end{aligned}
\end{equation*}
where $(a)$ holds by \cref{eq:L_sigma_r}, $(b)$ holds by union bound, and $(c)$ holds by Hoeffding inequality.
As a result, when $m \ge  \frac{9}{2\epsilon^2}\left(\ln\frac{2}{\delta} + \ln\Nn_{\infty}(\Ff, \frac{\epsilon}{3L})\right)$, we have
$\PP_{S\sim\Dd^m}\Big[\exists f\in \Ff, \Big|L_S(f) - L_\Dd(f)\Big| > \epsilon \Big] < \delta$.
% \begin{equation*}
% 	\begin{aligned}
% 		\PP_{S\sim\Dd^m}\Big[\exists f\in \Ff, \Big|L_S(f) - L_\Dd(f)\Big| > \epsilon \Big] < \delta
% 	\end{aligned}
% \end{equation*}
\end{proof}

\subsection{Proof of \cref{thm:CO}}\label{prf:thm:CO}
\thmCO*

We first provide an auxiliary lemma.
\begin{lemma}
    \label{lem:prf:thm:CO}
    For function class $\Ff$ and orbit averaging operator $\Xx$, if $\forall f,g \in \Ff, \ell_\infty(\Xx f, \Xx g) \le \ell_\infty(f, g)$, then $\Nn_\infty(\Xx \Ff, r) \le \Nn_\infty(\Ff, r)$ for any $r > 0$.
\end{lemma}
\begin{proof}
    $\forall r > 0$, Denote $\Ff_r$ as the smallest $r$-covering set  that covers $\Ff$ with size $\Nn_\infty(\Ff, r)$.
    $\forall f \in \Ff$, let $f_r \in \Ff_r$ be the function that $r$-covers $f$.
    We have $\ell_\infty(\Xx f_r, \Xx f) \le \ell_\infty(f_r, f) \le r$.
    Therefore, $\Xx\Ff_r$ is a $r$-covering set of $\Xx \Ff$, and we have $\Nn_\infty(\Xx \Ff, r) \le |\Xx\Ff_r| \le |\Ff_r| = \Nn_\infty$.
\end{proof}

\begin{proof}[Proof of \cref{thm:CO}]
For player $i \in [n]$ and $\forall \fNE, \gNE \in \FfNE$, assuming $\Uu$ is closed under any $\rho_i \in \Gg_i$. 
For $\Oo_i$, 
\begin{align*}
    l_{\infty}(\Oo_i \fNE,\Oo_i \gNE)
    =& \max_{u\in \Uu} \|\Oo_i \fNE(u) - \Oo_i \gNE(u)\|\\
    % =& \max_{u\in \Uu} \max_{j\in [n]} \|(\Oo_i \fNE)(u)_j - (\Oo_i \gNE)(u)_j\|\\
    =& \max_{j\in [n]} \max_{u\in \Uu} \|(\Oo_i \fNE)(u)_j - (\Oo_i \gNE)(u)_j\|\\
    % =& \max\Big\{ \max_{u\in\Uu} \|\fNE(u)_i - \gNE(u)_i \|,~\max_{j\ne i} \max_{u\in\Uu} \|(\Oo_i \fNE)(u)_j - (\Oo_i \gNE)(u)_j \| \Big\}\\
    =& \max\Big\{ \max_{u\in\Uu} \|\fNE(u)_i - \gNE(u)_i \|,~\max_{j\ne i} \max_{u\in\Uu} \|\frac{1}{|\Aa_i|!}\sum_{\rho_i \in \Gg_i} (\fNE(\rho_i u)_j - \gNE(\rho_i u)_j) \| \Big\}\\
    % \le& \max \Big\{ \max_{u\in\Uu} \|\fNE(u)_i - \gNE(u)_i \|,~\max_{j\ne i} \max_{u\in\Uu} \frac{1}{|\Aa_i|!}\sum_{\rho_i \in \Gg_i} \| \fNE(\rho_i u)_j - \gNE(\rho_i u)_j \| \Big\}\\
    \le& \max \Big\{ \max_{u\in\Uu} \|\fNE(u)_i - \gNE(u)_i \|,~\max_{j\ne i} \frac{1}{|\Aa_i|!}\sum_{\rho_i \in \Gg_i} \max_{u\in\Uu} \| \fNE(\rho_i u)_j - \gNE(\rho_i u)_j \| \Big\}\\
    =& \max \Big\{ \max_{u\in\Uu} \|\fNE(u)_i - \gNE(u)_i \|,~\max_{j\ne i} \frac{1}{|\Aa_i|!}\sum_{\rho_i \in \Gg_i} \max_{u\in\Uu} \| \fNE(u)_j - \gNE(u)_j \| \Big\}\\
    =& \max \Big\{ \max_{u\in\Uu} \|\fNE(u)_i - \gNE(u)_i \|,~\max_{j\ne i} \max_{u} \| \fNE(u)_j - \gNE(u)_j \| \Big\}\\
    % =& \max_{u\in\Uu} \max_{k\in[n]} \|\fNE(u)_k - \gNE(u)_k \|\\
    =& l_{\infty}(\fNE,\gNE)\\
\end{align*}
Since $\Oo = \Oo_1 \circ \cdots \circ \Oo_n$, we have 
\begin{equation}
    \label{eq:prf:thm:CO:O}    
    \ell_\infty(\Oo\fNE, \Oo\gNE) \le \ell_\infty(\fNE, \gNE).
\end{equation}

For $\Pp_i$,
\begin{align*}
    l_{\infty}(\Pp_i \fNE,\Pp_i \gNE)
    =& \max_{u \in \Uu} \max_{j\in [n]} \|(\Pp_i \fNE)(u)_j - (\Pp_i \gNE)(u)_j\|\\
    % =& \max\Big\{ \max_{j\ne i}\max_u \|(\Pp_i \fNE)(u)_j - (\Pp_i \gNE)(u)_j\|,~\max_u\|(\Pp_i \fNE)(u)_i - (\Pp_i \gNE)(u)_i\|\Big\}\\
    =& \max\Big\{ \max_{j\ne i}\max_u \|\fNE(u)_j - \gNE(u)_j\|,~\max_u\|\frac{1}{|\Aa_i|!}\sum_{\rho_i \in \Gg_i} \rho_i^{-1} (\fNE(\rho_i u)_i - \gNE(\rho_i u)_i)\|\Big\}\\
    =& \max\Big\{ \max_{j\ne i}\max_u \|\fNE(u)_j - \gNE(u)_j\|,~\max_u\|\frac{1}{|\Aa_i|!}\sum_{\rho_i \in \Gg_i} (\fNE(\rho_i u)_i - \gNE(\rho_i u)_i)\|\Big\}\\
    % \le& \max\Big\{ \max_{j\ne i}\max_u \|\fNE(u)_j - \gNE(u)_j\|,~\max_u \frac{1}{|\Aa_i|!}\sum_{\rho_i \in \Gg_i} \|\fNE(\rho_i u)_i - \gNE(\rho_i u)_i\|\Big\}\\
    \le& \max\Big\{ \max_{j\ne i}\max_u \|\fNE(u)_j - \gNE(u)_j\|,~\frac{1}{|\Aa_i|!}\sum_{\rho_i \in \Gg_i}\max_u  \|\fNE(\rho_i u)_i - \gNE(\rho_i u)_i\|\Big\}\\
    =& \max\Big\{ \max_{j\ne i}\max_u \|\fNE(u)_j - \gNE(u)_j\|,~\frac{1}{|\Aa_i|!}\sum_{\rho_i \in \Gg_i}\max_u  \|\fNE(u)_i - \gNE(u)_i\|\Big\}\\
    =& \max\Big\{ \max_{j\ne i}\max_u \|\fNE(u)_j - \gNE(u)_j\|,~\max_u  \|\fNE(u)_i - \gNE(u)_i\|\Big\}\\
    % =& \max_u \max_i \|\fNE(u)_i - \gNE(u)_i \|\\
    =& l_{\infty}(\fNE,\gNE)\\
\end{align*}
Since $\Pp = \Pp_1 \circ \cdots \circ \Pp_n$, we have 
\begin{equation}
    \label{eq:prf:thm:CO:P}    
    \ell_\infty(\Pp\fNE, \Pp\gNE) \le \ell_\infty(\fNE, \gNE).
\end{equation}

For CE or CCE approximator $\fcCE \in \FfcCE$ and $\Qq_i$, we have
\begin{align*}
    l_\infty(\Qq_i \fcCE,\Qq_i \gcCE)
    =& \max_{u\in\Uu} \| (\Qq_i \fcCE)(u) - (\Qq_i \gcCE)(u) \|
    \\
    =& \max_u \| \frac{1}{|\Aa_i|!}\sum_{\rho_i\in\Gg_i} \rho_i^{-1} (\fcCE(\rho_i u) - \gcCE(\rho_i u))\|
    \\
    \le& \max_u \frac{1}{|\Aa_i|!}\sum_{\rho_i\in\Gg_i} \| \rho_i^{-1} (\fcCE(\rho_i u) - \gcCE(\rho_i u))\|\\
    \le& \frac{1}{|\Aa_i|!}\sum_{\rho_i\in\Gg_i} \max_u \| \rho_i^{-1} (\fcCE(\rho_i u) - \gcCE(\rho_i u))\|\\
    =& \frac{1}{|\Aa_i|!}\sum_{\rho_i\in\Gg_i} \max_u \| \fcCE(\rho_i u) - \gcCE(\rho_i u)\|
    \\
    =& \frac{1}{|\Aa_i|!}\sum_{\rho_i\in\Gg_i} \max_u \| \fcCE(u) - \gcCE(u)\|
    \\
    =& l_\infty (\fcCE,\gcCE)\\
\end{align*}
Since $\Qq = \Qq_1 \circ \cdots \circ \Qq_n$, we have 
\begin{equation}
    \label{eq:prf:thm:CO:Q}    
    \ell_\infty(\Qq\fcCE, \Qq\gcCE) \le \ell_\infty(\fcCE, \gcCE).
\end{equation}

Combing \cref{lem:prf:thm:CO}, \cref{eq:prf:thm:CO:O}, \cref{eq:prf:thm:CO:P} and \cref{eq:prf:thm:CO:Q}, we finish the proof.
\end{proof}

\subsection{Proof of \cref{thm:CCE:E}}\label{prf:thm:CCE:E}
\thmCCEE*

We first prove a lemma about the property of $\Ee_i(\pi,u)$ and $\Ee_i^\CE(\pi,u)$.
\begin{lemma}
    $\Ee_i(\pi,u)$ and $\Ee_i^\CE(\pi,u)$ are
    convex on $\pi$, i.e.
    \begin{equation*}
        \begin{aligned}
            p \Ee_i^{\text{(CE)}} (\pi_1,u) + (1-p) \Ee_i^{\text{(CE)}}(\pi_2,u) \ge \Ee_i^{\text{(CE)}} (p\pi_1 + (1-p)\pi_2,u),\quad \forall p\in[0,1]
        \end{aligned}
    \end{equation*}
\end{lemma}

\begin{proof}
We recall the definition $\Ee_i(\pi,u) = \max_{a_i\in \Aa_i} u_i(a_i,\pi_{-i}) - u_i(\pi)$ for CCE approximator
and $\Ee_i^\CE(\pi,u) = \max_{\phi_i \in \Aa_i \to \Aa_i} \sum_a \pi(a) u_i(\phi_i(a_i),a_{-i}) - u_i(\pi)$ for CE approximator.
$u_i(a_i,\pi_{-i})$ is linear on $\pi$.
% \begin{align*}
%      u_i(a_i,\pi_{-i})
%     =& \sum_{b_i\in\Aa_i,b_{-i}\in\Aa_{-i}} u_i(b_i,b_{-i})\cdot(a_i, \pi_{-i})(b_i,b_{-i})
%     \\
%     =& \sum_{b_{-i}\in\Aa_{-i}} u_i(a_i,b_{-i})\cdot(a_i, \pi_{-i})(a_i,b_{-i})
%     \\
%     =& \sum_{b_{-i}\in\Aa_{-i}} u_i(a_i,b_{-i})\cdot\pi_{-i}(b_{-i})
%     \\
%     =& \sum_{b_i\in\Aa_i,b_{-i}\in\Aa_{-i}} u_i(a_i,b_{-i})\cdot\pi(b_i,b_{-i})
% \end{align*}
% Then,
% \begin{align*}
%     p u_i(a_i,\pi_{-i}^{1}) + (1-p) u_i(a_i,\pi_{-i}^{2})
%     =& \sum_{b_i\in\Aa_i,b_{-i}\in\Aa_{-i}} u_i(a_i,b_{-i})\cdot(p \pi^{1}(b_i,b_{-i}) + (1-p) \pi^{2}(b_i,b_{-i}) )
%     \\
%     =& \sum_{b_i\in\Aa_i,b_{-i}\in\Aa_{-i}} u_i(a_i,b_{-i})\cdot(p \pi^{1} + (1-p)\pi^{2})(b_i,b_{-i})
%     \\
%     =& u_i(a_i,p \pi^{1} + (1-p)\pi^{2})
% \end{align*}
Given $\phi$, $\sum_a \pi(a) u_i(\phi_i(a_i),a_{-i})$ is also linear on $\pi$.
Moreover, the maximum operator on a set of linear functions will induce a convex function.

\end{proof}

\begin{proof}[Proof of \cref{thm:CCE:E}]
For $f\in\FfcCE$ and $\forall i,j \in [n]$,
\begin{align*}
    \EE_{u\sim\Dd}[\Ee_i^\text{(CE)}(\Qq_j f(u),u)]
    =& \EE_{u\sim \Dd}[\Ee_i^\text{(CE)}(\frac{1}{|\Aa_j|!}\sum_{\rho_j\in\Gg_j} \rho_j^{-1} f(\rho_j u),u)]& &,\text{by definition}\\
    \le& \frac{1}{|\Aa_j|!}\sum_{\rho_j\in\Gg_j} \EE_{u\sim \Dd}[\Ee_i^\text{(CE)}(\rho_j^{-1}f(\rho_j u),u)]& &,\text{by convexity}\\
    =& \frac{1}{|\Aa_j|!}\sum_{\rho_j\in\Gg_j} \EE_{v\sim \Dd}[\Ee_i^\text{(CE)}(\rho_j^{-1}f(v),\rho_j^{-1} v)]& &,\text{let $v=\rho_j u$}\\
    =& \frac{1}{|\Aa_j|!}\sum_{\rho_j\in\Gg_j} \EE_{v\sim \Dd}[\Ee_i^\text{(CE)}(f(v),v)]& &,\text{invariance of $\Ee_i^\text{(CE)}(\pi,u)$ under $\rho_j^{-1}\in\Gg_j$}\\
    =& \EE_{u\sim \Dd}[\Ee_i^\text{(CE)}(f(u),u)]\\
\end{align*}

Since $\Qq = \circ_{i} \Qq_i$ and $\Ee = \max_i \Ee_i$, we have
\begin{equation*}
    \EE_{u\sim\Dd}[\Ee(\Qq f(u),u)]\le \EE_{u\sim\Dd}[\Ee(f(u),u)]
\end{equation*}
\end{proof}

% -----
\subsection{Proof of \cref{thm:NE:Bi:c}}\label{prf:thm:NE:Bi:c}
\thmNEBiC*

\begin{proof}
We only prove for the $\Pp$-projected case; the proof for $\Oo$-projected case is similar and therefore omitted.

Recall
\begin{align*}
    \Ee_i(\sigma,u) = \max_{a_i\in\Aa_i} u_i(a_i,\sigma_{-i}) - u_i(\sigma)
\end{align*}

Denote $u_1(\sigma) + u_2(\sigma) \equiv c$, then 
\begin{align*}
    \sum_i \Ee_i(\sigma,u) = \max_{a_1\in\Aa_1,a_2\in\Aa_2} u_1(a_1,\sigma_2) + u_2(a_2,\sigma_1) - c
\end{align*}
Then we have
\begin{align*}
    \EE_{u\sim\Dd} [\sum_i \Ee_i((\Pp\fNE)(u),u)]
    =& \EE_{u\sim\Dd} [\max_{a_1,a_2} u_1(a_1,(\Pp\fNE)(u)_2) + u_2(a_2,(\Pp\fNE)(u)_1) - c]
    \\
    =& \EE_{u\sim\Dd} [\max_{a_1} u_1(a_1,(\Pp\fNE)(u)_2)] + \EE_{u\sim\Dd} [\max_{a_2} u_2(a_2,(\Pp\fNE)(u)_1)] - c
\end{align*}

For the first term,
\begin{align*}
    \EE_{u\sim\Dd} [\max_{a_1} u_1(a_1,(\Pp\fNE)(u)_2)]
    =& \EE_{u\sim\Dd} [\max_{a_1} u_1(a_1,\frac{1}{|\Aa_2|!}\sum_{\rho_2\in\Gg_2}\rho_2^{-1} \fNE(\rho_2 u)_2)]
    \\
    \le& \frac{1}{|\Aa_2|!}\sum_{\rho_2\in\Gg_2} \EE_{u\sim\Dd} [\max_{a_1} u_1(a_1,\rho_2^{-1} \fNE(\rho_2 u)_2)]
    \\
    =& \frac{1}{|\Aa_2|!}\sum_{\rho_2\in\Gg_2} \EE_{v\sim\Dd} [\max_{a_1} (\rho_2^{-1} v)_1(a_1,\rho_2^{-1} \fNE(v)_2)]
    \\
    =& \frac{1}{|\Aa_2|!}\sum_{\rho_2\in\Gg_2} \EE_{v\sim\Dd} [\max_{a_1} v_1(a_1, \fNE(v)_2)]
    \\
    =& \EE_{u\sim\Dd} [\max_{a_1} u_1(a_1, \fNE(u)_2)]
\end{align*}

Similarly, for the second term,
\begin{align*}
    & \EE_{u\sim\Dd} [\max_{a_2} u_2(a_2,(\Pp\fNE)(u)_1)]
    \le \EE_{u\sim\Dd} [\max_{a_2} u_2(a_2, \fNE(u)_1)]
\end{align*}

Above all,
\begin{align*}
    \EE_{u\sim\Dd} [\sum_i \Ee_i((\Pp\fNE)(u),u)]
    =& \EE_{u\sim\Dd} [\max_{a_1} u_1(a_1,(\Pp\fNE)(u)_2)] + \EE_{u\sim\Dd} [\max_{a_2} u_2(a_2,(\Pp\fNE)(u)_1)] - c
    \\
    \le& \EE_{u\sim\Dd} [\max_{a_1} u_1(a_1,\fNE(u)_2)] + \EE_{u\sim\Dd} [\max_{a_2} u_2(a_2,\fNE(u)_1)] - c
    \\
    =& \EE_{u\sim\Dd} [\sum_i \Ee_i(\fNE(u),u)]
\end{align*}
\end{proof}
% -----

% \subsection{Proof of \cref{thm:NEG}}
% \label{prf:thm:NEG}

% Consider a 2-player cooperation game with identity utility, i.e. the game matrix $u_i$ is 
% \begin{table}[H]
%     \centering
%     \begin{tabular}{|c|c|}
%         \hline
%         1 & 0 \\
%         \hline
%         0 & 1 \\
%         \hline
%     \end{tabular}
%     % \caption*{Caption}
%     % \label{tab:game}
% \end{table}
% for $i\in\{1,2\}$

% Let $f$ be a NE approximator which holds PE and PI, $f(u)=(\sigma_1,\sigma_2)$, where $\sigma_1 = (p_1,1-p_1)$ and $\sigma_2=(p_2,1-p_2)$.

% Denote $\Gg_i = \{e,\rho_i\}$ be the permutation group of player i, where $i\in\{1,2\}$, $e$ is the identity function, and $\rho_i$ is the swap function for player i. Note that $\rho_1\rho_2 u = u$. $f(u)_1 = f(\rho_1 \rho_2 u)_1 = \rho_1 f(\rho_2 u)_1 = \rho_1 f(u)_1 = (1-p_1,p_1)$, then we know $p_1 = 0.5$. Similarly, we have $p_2 = 0.5$. Then the only NE that $f$ could find is $\sigma_1=(0.5,0.5)$ and $\sigma_2=(0.5,0.5)$.

% However, this game has three NE: $\sigma_1=\sigma_2=(1,0)$, $\sigma_1=\sigma_2=(0,1)$, and $\sigma_1=\sigma_2=(0.5,0.5)$. The first two NEs induce social welfare to be 2, while the last NE induces social welfare to be 1.

% -----

\subsection{Proof of \cref{thm:NE:P}}\label{prf:thm:NE:P}
\thmNEP*

We first introduce a useful lemma. It is about the property of $\Ee_i(\sigma,u)$
\begin{lemma}
\label{lem:NE:E}
$\Ee_i(\sigma,u)$ is
\begin{enumerate}
    \item Linear on $\sigma_i$, i.e. 
    \begin{equation*}
        p \Ee_i((\sigma_i^{1},\sigma_{-i}),u) + (1-p) \Ee_i((\sigma_i^{2},\sigma_{-i}),u) = \Ee_i((p\sigma_i^{1}+(1-p)\sigma_i^2,\sigma_{-i}),u),~ \forall p\in[0,1]
    \end{equation*}
    
    \item Convex on $\sigma_j$, i.e.
    \begin{equation*}
        p \Ee_i((\sigma_j^{1},\sigma_{-j}),u) + (1-p) \Ee_i((\sigma_j^{2},\sigma_{-j}),u) \ge \Ee_i((p\sigma_j^{1}+(1-p)\sigma_j^2,\sigma_{-j}),u),~ \forall p\in[0,1], j\ne i
    \end{equation*}
    
    % \item Invariant under $\rho_j \in\Gg_j$, i.e.
    % \begin{equation*}
    %     \Ee_i(\sigma,u) = \Ee_i(\rho_k \sigma, \rho_k u),\quad \forall k\in[n]
    % \end{equation*}
\end{enumerate}
\end{lemma}

\begin{proof}
\label{prf:lem:NE:E}
    % The first and the second claim are trivial.
    We recall the definition $\Ee_i(\sigma,u) = \max_{a_i\in\Aa_i} u_i(a_i,\sigma_{-i}) - u_i(\sigma)$.
    Notice that $u_i(\sigma)$ is linear on $\sigma_k$ for all $k\in [n]$, thus both $u_i(a_i,\sigma_{-i})$ and $u_i(\sigma)$ are linear on $\sigma_k$ for any $k \in [n]$. 
    Moreover, the maximum operator on a set of linear functions will induce a convex function.

    % For the third claim, we verify it by definition. For player $i$, we have
    % \begin{align*}
    %     \Ee_i(\rho_i \sigma,\rho_i u)
    %     =& \max_{a_i \in \Aa_i} \rho_i u_i(a_i,\rho_i\sigma_{-i}) - \rho_i u_i(\rho_i\sigma)
    %     = \max_{a_i \in \Aa_i} \rho_i u_i(a_i,\sigma_{-i}) - \rho_i u_i(\rho_i\sigma_i,\sigma_{-i})
    %     \\
    %     =& \max_{a_i \in \Aa_i} u_i(\rho_i^{-1}a_i ,\sigma_{-i}) - u_i(\rho_i^{-1}\rho_i\sigma_i,\sigma_{-i})
    %     \overset{(a)}{=} \max_{a_i\in\Aa_i} u_i(a_i,\sigma_{-i}) - u_i(\sigma_i,\sigma_{-i}) 
    %     = \Ee_i(\sigma,u),
    % \end{align*}
    % where  $(a)$ holds since $\rho_i$ is a bijection on $\Aa_i$.
    % For player $j \neq i$, we have
    % \begin{align*}
    %     \Ee_i(\rho_j\sigma,\rho_j u)
    %     =& \max_{a_i\in\Aa} \rho_j u_i(a_i,\rho_j\sigma_{-i}) - \rho_j u_i(\rho_j\sigma)
    %     = \max_{a_i\in\Aa_i} u_i(a_i,\rho_j^{-1}\rho_j \sigma_{-i}) - u_i(\rho_j^{-1}\rho_j \sigma)
    %     \\
    %     =& \max_{a_i\in\Aa_i} u_i(a_i,\sigma_{-i}) - u_i(\sigma)
    %     = \Ee_i(\sigma,u)
    % \end{align*}
\end{proof}

\begin{proof}[Proof of \cref{thm:NE:P}]
We prove the theorem in two steps.

\paragraph{Step 1}
First, we show that 
\begin{equation*}
    \EE_{u \sim \Dd}[\Ee_i((\Pp_i \fNE)(u), u)] = \EE_{u \sim \Dd}[\Ee_i(\fNE(u), u)],\quad \forall \fNE \in \FfNE
\end{equation*}

By definition,
\begin{align*}
    &\EE_{u \sim \Dd}[\Ee_i(\Pp_i \fNE(u), u)]
    \\
    =& \EE_{u \sim \Dd}[\Ee_i((\frac{1}{|\Aa_i |!}\sum_{\rho_i \in \Gg_i} \rho_i^{-1} f(\rho_i u)_i , f(u)_{-i}),u)]
    \\
    =& \frac{1}{|\Aa_i |!}\sum_{\rho_i \in \Gg_i} \EE_{u \sim \Dd}[\Ee_i(( \rho_i^{-1} f(\rho_i u)_i , f(u)_{-i}),u)]& & ,\text{by linearity of $\Ee_i(\sigma,u)$ on $\sigma_i$}
    \\
    =& \frac{1}{|\Aa_i |!}\sum_{\rho_i \in \Gg_i} \EE_{v \sim \Dd}[\Ee_i(( \rho_i^{-1} f(v)_i , f(\rho_i^{-1} v)_{-i}),\rho_i^{-1} v)]& & ,\text{let $v=\rho_i u$ and use the invariance of $\Dd$}
    \\
    =& \frac{1}{|\Aa_i |!}\sum_{\rho_i \in \Gg_i} \EE_{v \sim \Dd}[\Ee_i(( \rho_i^{-1} f(v)_i , f(v)_{-i}),\rho_i^{-1} v)]& & ,\text{OPI of $f$}
    \\
    =& \frac{1}{|\Aa_i |!}\sum_{\rho_i \in \Gg_i} \EE_{u \sim \Dd}[\Ee_i(( f(u)_i , f(u)_{-i}),u)]& & ,\text{invariance of $\Ee_i(\sigma,u)$ under $\rho_i^{-1}\in\Gg_i$}
    \\
    =& \EE_{u \sim \Dd}[\Ee_i(\fNE(u),u)]
\end{align*}

\paragraph{Step 2}
Then we show that
\begin{equation*}
    \EE_{u \sim \Dd}[\Ee_j((\Pp_i \fNE)(u), u)] \le \EE_{u \sim \Dd}[\Ee_j(\fNE(u), u)],\quad \forall \fNE \in \FfNE, j\ne i
\end{equation*}

\begin{align*}
    & \EE_{u \sim \Dd}[\Ee_j((\Pp_i \fNE)(u), u)]\\
    =& \EE_{u \sim \Dd}[\Ee_j((\frac{1}{|\Aa_i |!}\sum_{\rho_i \in \Gg_i} \rho_i^{-1} f(\rho_i u)_i , f(u)_{-i}),u)]\\
    \le& \frac{1}{|\Aa_i |!}\sum_{\rho_i \in \Gg_i} \EE_{u \sim \Dd}[\Ee_j(( \rho_i^{-1} f(\rho_i u)_i , f(u)_{-i}),u)]& &,\text{by convexity of $\Ee_j(\sigma,u)$ on $\sigma_i$}\\
    =& \frac{1}{|\Aa_i |!}\sum_{\rho_i \in \Gg_i} \EE_{v \sim \Dd}[\Ee_j(( \rho_i^{-1} f(v)_i , f(\rho_i^{-1} v)_{-i}),\rho_i^{-1} v)]& &,\text{let $v=\rho_i u$ and use the invariance of $\Dd$}\\
    =& \frac{1}{|\Aa_i |!}\sum_{\rho_i \in \Gg_i} \EE_{v \sim \Dd}[\Ee_j(( \rho_i^{-1} f(v)_i , f(v)_{-i}),\rho_i^{-1} v)]& &,\text{OPI of $f$}\\
    =& \frac{1}{|\Aa_i |!}\sum_{\rho_i \in \Gg_i} \EE_{u \sim \Dd}[\Ee_j(( f(u)_i , f(u)_{-i}),u)]& &,\text{invariance of $\Ee_j(\sigma,u)$ under $\rho_i^{-1}\in\Gg_i$}\\
    =& \EE_{u \sim \Dd}[\Ee_j(\fNE(u),u)]\\
\end{align*}

Since $\Pp = \circ_{i} \Pp_i$ and $\Ee = \max_i \Ee_i$, we have 
\begin{equation*}
    \EE_{u \sim \Dd}[\Ee((\Pp \fNE)(u), u)] \le
    \EE_{u \sim \Dd}[\Ee(\fNE(u), u)]
\end{equation*}

\end{proof}
% -----

\subsection{Proof of \cref{thm:NE:O}}\label{prf:thm:NE:O}
\thmNEO* 

\begin{proof}
We prove the theorem in two steps, similar to the proof of \cref{thm:NE:P}.

\paragraph{Step 1}
First we show that for player $i\in\{1,2\}$, let $\{j\}=\{1,2\} \backslash \{i\}$,
\begin{equation*}
    \EE_{u\sim\Dd} [\Ee_i((\Oo_i \fNE)(u),u)] \le \EE_{u\sim\Dd} [\Ee_i(\fNE(u),u)]
\end{equation*}
This is because
\begin{align*}
    \EE_{u\sim\Dd} [\Ee_i((\Oo_i \fNE)(u),u)]
    =& \EE_{u\sim\Dd} [\Ee_i((\fNE(u)_i,\frac{1}{|\Aa_i|!}\sum_{\rho_i \in\Gg_i}\fNE(\rho_i u)_j),u)]
    \\
    \le& \frac{1}{|\Aa_i|!}\sum_{\rho_i \in\Gg_i}\EE_{u\sim\Dd} [\Ee_i((\fNE(u)_i,\fNE(\rho_i u)_j),u)]& &,\text{by convexity of $\Ee_i$ on $\sigma_j$}
    \\
    =& \frac{1}{|\Aa_i|!}\sum_{\rho_i \in\Gg_i}\EE_{v\sim\Dd} [\Ee_i((\fNE(\rho_i^{-1} v)_i,\fNE(v)_j),\rho_i^{-1} v)]& &,\text{let $v=\rho_i u$}
    \\
    =& \frac{1}{|\Aa_i|!}\sum_{\rho_i \in\Gg_i}\EE_{v\sim\Dd} [\Ee_i((\rho_i^{-1} \fNE(v)_i,\fNE(v)_j),\rho_i^{-1} v)]& &,\text{by PPE of $\fNE$}
    \\
    =& \frac{1}{|\Aa_i|!}\sum_{\rho_i \in\Gg_i}\EE_{v\sim\Dd} [\Ee_i((\fNE(v)_i,\fNE(v)_j),v)]& &,\text{invariance of $\Ee_i(\sigma,u)$ under $\rho_i^{-1}\in\Gg$}
    \\
    =& \EE_{u\sim\Dd} [\Ee_i((\fNE)(u),u)]
\end{align*}

\paragraph{Step 2}
Then we show that if $j\ne i$ and $\{i,j\}=\{1,2\}$
\begin{equation*}
    \EE_{u\sim\Dd} [\Ee_j((\Oo_i \fNE)(u),u)] = \EE_{u\sim\Dd} [\Ee_j(\fNE(u),u)]
\end{equation*}
This is because
\begin{align*}
    \EE_{u\sim\Dd} [\Ee_j((\Oo_i \fNE)(u),u)]
    =& \EE_{u\sim\Dd} [\Ee_j((\fNE(u)_i,\frac{1}{|\Aa_i|!}\sum_{\rho_i\in\Gg_i} \fNE(\rho_i u)_j),u)]
    \\
    =& \frac{1}{|\Aa_i|!}\sum_{\rho_i\in\Gg_i} \EE_{u\sim\Dd} [\Ee_j((\fNE(u)_i, \fNE(\rho_i u)_j),u)]& &,\text{by linearity of $\Ee_j$ on $\sigma_j$}
    \\
    =& \frac{1}{|\Aa_i|!}\sum_{\rho_i\in\Gg_i} \EE_{v\sim\Dd} [\Ee_j((\fNE(\rho_i^{-1} v)_i, \fNE(v)_j),\rho_i^{-1} v)]& &,\text{let $v=\rho_i u$}
    \\
    =& \frac{1}{|\Aa_i|!}\sum_{\rho_i\in\Gg_i} \EE_{v\sim\Dd} [\Ee_j((\rho_i^{-1} \fNE(v)_i, \fNE(v)_j),\rho_i^{-1} v)]& &,\text{by PPE of $\fNE$}
    \\
    =& \frac{1}{|\Aa_i|!}\sum_{\rho_i\in\Gg_i} \EE_{v\sim\Dd} [\Ee_j((\fNE(v)_i, \fNE(v)_j),v)]& &,\text{invariance of $\Ee_j(\sigma,u)$ under $\rho_i^{-1}\in\Gg_i$}
    \\
    =& \EE_{u\sim\Dd} [\Ee_j(\fNE(u), u)]
\end{align*}

Since $\Oo = \circ_{i} \Oo_i$ and $\Ee = \max_i \Ee_i$, we have
\begin{equation*}
    \EE_{u\sim\Dd}[\Ee(\Oo \fNE(u),u)]\le \EE_{u\sim\Dd}[\Ee(\fNE(u),u)]
\end{equation*}
\end{proof}
% -----

\section{Omitted Proofs in \cref{sec:limitation}}
\subsection{Proof of \cref{thm:NEG}}\label{prf:thm:NEG}
\thmNEG*
\begin{proof}
Let $f$ be a PPE and OPI NE approximator.
Denote $f(u)=(\sigma_i)_{i\in[n]}$. 
For player $k$ that $a^*_k \in V(\rho_k)$, we get 
\begin{equation}
\label{eq:prf:thm:NEG}
    \sigma_k = f(u)_k
    \overset{(a)}{=} f(\rho u)_k
    \overset{(b)}{=} f(\rho_k u)_k
    \overset{(c)}{=} \rho_k f(u)_k
    = \rho_k \sigma_k,
\end{equation}
where $(a)$ holds since $u$ is permutable w.r.t. $\rho$, $(b)$ holds by OPI of $f$, and $(c)$ holds by PPE of $f$.
If $a^*$ can be found by $f$, we will have 
$1 = \sigma_k(a^*_k) \overset{(d)}{=} \rho_k \sigma_k(a^*_k) = \sigma_k(\rho_k^{-1} (a^*_k))$, where $(d)$ holds by \cref{eq:prf:thm:NEG}.
However, such result leads to a contradiction, because $a^*_k \neq \rho_k^{-1}(a_k)$ but $\sigma_k(a^*_k)=\sigma(\rho_k^{-1}(a^*_k))=1$.

Let $f$ be a PE (C)CE approximator. Denote $f(u)=\pi$, we have
\begin{equation}
\label{eq:prf:thm:NEG2}
    \pi = f(u)
    \overset{(a)}{=} f(\rho u)
    \overset{(b)}{=} \rho f(u)
    = \rho \pi
\end{equation}
where $(a)$ holds since $u$ is permutable w.r.t. $\rho$, $(b)$ holds by PE of $f$. If $a^*$ can be found by $f$, we will have
$1 = \pi(a^*) \overset{(c)}{=} \rho\pi(a^*) = \pi(\rho^{-1} a^*) = \pi(\rho_1^{-1} a_1^*,\cdots,\rho_n^{-1} a_n^*)$, where $(c)$ holds by \cref{eq:prf:thm:NEG2}.
However, from $a_k^*\in V(\rho_k)$ we know $\rho_k^{-1}(a_k^*)\ne a_k^*$, then $\rho^{-1} a^* \ne a^*$, but $\pi(a^*) = \pi(\rho^{-1} a^*) = 1$.
\end{proof}

\subsection{Proof of \cref{thm:NEG:ratio:CCE}}\label{prf:thm:NEG:ratio:CCE}
\thmNEGRatioCCE*
\begin{proof}
Assume $f\in\Ff^\text{(C)CE}_\mathrm{general}$ is an (C)CE approximator that always finds the strategy that maximizes the social welfare.
Afterward, we construct another $f_0$ that satisfies PE and always finds the strategy that maximizes social welfare.
$f_0$ is constructed by orbit averaging:
\begin{align*}
    f_0(u) = \Qq f(u),
\end{align*}
thus $f_0$ is PE.

Denote $\Dd$ as an arbitrary payoff distribution of $u$ such that $\Dd$ is invariant under permutation and the cardinality of its support is finite.
We have
% then we show that
% \begin{align*}
%     \EE_{u\sim\Dd} \SW(f_0(u),u) = \EE_{u\sim\Dd}\SW(f(u),u)
% \end{align*}
\begin{align*}
    \EE_{u\sim\Dd} \SW(\Qq_i f(u),u)
    =& \EE_{u\sim\Dd} \SW(\frac{1}{|\Aa_i|!}\sum_{\rho_i\in\Gg_i} \rho_i^{-1} f(\rho_i u),u)
    \\
    =& \EE_{u\sim\Dd} \sum_{i=1}^n u_i(\frac{1}{|\Aa_i|!}\sum_{\rho_i\in\Gg_i} \rho_i^{-1} f(\rho_i u))
    \\
    =& \frac{1}{|\Aa_i|!}\sum_{\rho_i\in\Gg_i} \EE_{u\sim\Dd} \sum_{i=1}^n u_i( \rho_i^{-1} f(\rho_i u))
    \\
    =& \frac{1}{|\Aa_i|!}\sum_{\rho_i\in\Gg_i} \EE_{v\sim\Dd} \sum_{i=1}^n (\rho_i^{-1}v)_i( \rho_i^{-1} f(v)) & &,\text{let } v = \rho_i u
    \\
    =& \frac{1}{|\Aa_i|!}\sum_{\rho_i\in\Gg_i} \EE_{v\sim\Dd} \sum_{i=1}^n v_i(f(v))
    \\
    =& \EE_{u\sim\Dd} \sum_{i=1}^n u_i(f(u))\\
    =& \EE_{u\sim\Dd} \SW(f(u),u)
\end{align*}

Due to that $\Qq = \Qq_1\circ\cdots\circ\Qq_n$, we have
\begin{align*}
    \EE_{u\sim\Dd} \SW(f_0(u),u) = \EE_{u\sim\Dd}\SW(f(u),u)
\end{align*}

Due to the arbitrariness of $\Dd$, we know that $f_0$  maximizes the social welfare w.r.t. any $u$.

From above, we immediately know 
\begin{align*}
    \SWR({\Ff}^\mathrm{(C)CE}_\mathrm{PE}, \FfcCE_\mathrm{general}) = 1
\end{align*}

\end{proof}
% -----

\subsection{Proof of \cref{thm:NEG:ratio:NE}}\label{prf:thm:NEG:ratio:NE}
\thmNEGRatioNE*

\subsubsection{Proof of \cref{eq:NEG:ratio:OPI} and \cref{eq:NEG:ratio:both}}
We first prove the theorem with respect to $\Ff^\mathrm{NE}_\mathrm{OPI}$ and $\Ff^\mathrm{NE}_\mathrm{both}$

\paragraph{Step 1}
On the one part, we prove
\begin{align*}
\left.
\begin{aligned}
    \SWR({\Ff}^\NE_\mathrm{OPI},\Ff^\NE_\mathrm{general})
    \\
    \SWR({\Ff}^\NE_\mathrm{both},\Ff^\NE_\mathrm{general})
\end{aligned}
\right\}\le \frac{1}{M^{N-1}}
\end{align*}

We prove this by construction.

Consider a game with $N$ player and $\Aa_i=[M]$ for $i\in[N]$.
$\forall a \in \Aa, i \in [N]$, define the payoff $\bar u$ as follows: 
\begin{equation*}
    \bar u_i(a) =     \begin{cases}
        1 &,\text{if }a_1 = a_2 = \dots = a_N
        \\
        0 &,\text{otherwise}
    \end{cases}
\end{equation*}
% If all players play the same action, then they all receive a utility of $1$, otherwise, they get $0$ utility.
Define $U=\{u'|u' = \circ_{i}\rho_i \bar u,\rho_i\in\Gg_i\}$ and $\Dd$ as a uniform distribution on $U$. Easy to certify that $\Dd$ is a permutation-invariant distribution.

Let $\Tilde{f}\in \Tilde\Ff^\NE_\mathrm{general}$ be the NE oracle that $\Tilde{f}(\bar u)_i = 1$ and for any $u'=\circ_{i} \rho_i \bar u\in U$, $\Tilde{f}(u')_i=\rho_i (1)$. 
% If $u'$ has multiple representations, the oracle will choose an arbitrary one. 
Intuitively, the oracle will choose the action that will provide all players with revenue $1$, leading to a social welfare of $N$. 
Since each player has got her maximum possible utility, we have 
\begin{equation}\label{eq:maxSW:oracle}
    \max_{f\in\Ff^\NE_\mathrm{general}} \EE_{u\sim\Dd} \SW(f(u),u)=\max_{\Tilde{f}\in\widetilde{\Ff}^\NE_\mathrm{general}} \EE_{u\sim\Dd} \SW(\Tilde{f}(u),u)=N.
\end{equation}

For any $j_1,j_2\in[M]$ and $j_1 < j_2$, let $\rho^{(j_1,j_2)}_{i} = (1, \dots, j_2, \dots, j_1, \dots, M)$ for all player $i \in [N]$ be the swap permutation that swaps actions $j_1$ and $j_2$ and keeps other actions still. 
Then $\circ_{i\ne j} \rho_i^{(j_1,j_2)} \bar u = \rho_j^{(j_1,j_2)} \bar u$ for player $j$. 
For $f \in \FfNE_\mathrm{OPI}$, we have $f(\bar u)_j = f(\circ_{i\ne j} \rho_i^{(j_1,j_2)} \bar u)_j = f(\rho_j^{(j_1,j_2)} \bar u)_j$ for arbitrary swap permutation $\rho_j^{(j_1,j_2)}$. 
Since any permutation can be achieved by composition of swap permutations, we have $\forall \rho_j\in\Gg_j$, $f(\bar u)_j = f(\rho_j \bar u)_j$.
Based on that, and by OPI of $f$, $\forall \rho = \circ_{i\in [N]} \rho_i$ we have $f(\bar u)_j = f(\rho \bar u)_j$, i.e. $f$ is a constant function on $U$.
Without loss of generality, we denote $f(u)\equiv\sigma$ for all $u\in U$.
Then 
\begin{equation*}
    \EE_{u\sim \Dd}\SW(f(u),u) 
    = \frac{1}{|U|}\sum_{u'\in U}\SW(\sigma,u') 
    = \frac{1}{(M!)^{N-1}}\SW(\sigma,\sum_{u'\in U}u').
\end{equation*}
Additionally,  we have  $(\sum_{u'\in U}u')(a) = ((M-1)!)^{N-1}$ for any $a \in \Aa$. 
Based on that, we have
\begin{equation}\label{eq:maxSW:OPI}
    \max_{f\in\Ff^\mathrm{NE}_\mathrm{OPI}} \EE_{u\sim \Dd}\SW(f(u),u)
    = \frac{1}{(M!)^{N-1}} \cdot N((M-1)!)^{N-1}
    =\frac{N}{M^{N-1}}.
\end{equation}
Combining \cref{eq:maxSW:oracle} and \cref{eq:maxSW:OPI}, we have
\begin{equation*}
    \SWR(\Ff^{\mathrm{NE}}_{\mathrm{OPI}},\Ff^\NE_\mathrm{general})\le\frac{1}{M^{N-1}}.
\end{equation*}

Due to $\Ff^\mathrm{NE}_\mathrm{both}\subseteq\Ff^\mathrm{NE}_\mathrm{OPI}$, we immediately know
\begin{align*}
    \SWR(\Ff^{\mathrm{NE}}_{\mathrm{both}},\Ff^\NE_\mathrm{general})\le\frac{1}{M^{N-1}}
\end{align*}

\paragraph{Step 2}
On the other part, we prove
\begin{align*}
\left.
\begin{aligned}
    \SWR({\Ff}^\NE_\mathrm{OPI},\Ff^\NE_\mathrm{general})
    \\
    \SWR({\Ff}^\NE_\mathrm{both},\Ff^\NE_\mathrm{general})
\end{aligned}
\right\}\ge 1/M^{N-1}
\end{align*}

Define the maximum possible utility (MPU) for player $i$ with respect to utility $u_i$ and action $a_i$ as
\begin{equation}
    \mathrm{MPU}(u_i,a_i) \coloneqq \max_{a_{-i}\in\Aa_{-i}} u_i(a_i,a_{-i})
\end{equation}

Define the set of maximum possible utility best response for player $i$ w.r.t. $u_i$ as 
\begin{equation*}
    \Bb_i(u_i) \coloneqq \{a_i\in\Aa_i : \mathrm{MPU}(u_i,a_i) = \max_{a'_i\in\Aa_i}\mathrm{MPU}(u_i,a'_i) \}
\end{equation*}

We first conduct some simplification to the target. 
\begin{align*}
    & \SWR({\Ff}^\NE_\mathrm{both},\Ff^\NE_\mathrm{general})
    = \inf_\Dd \frac{\max_{f\in\Ff^\mathrm{NE}_\mathrm{both}} \EE_{u\sim \Dd}\SW(f(u),u)}{\max_{f\in{\Ff^\NE_\mathrm{general}}} \EE_{u\sim\Dd} \SW(f(u),u)}
    \ge \inf_\Dd \frac{\max_{f\in\Ff^\mathrm{NE}_\mathrm{both}} \EE_{u\sim \Dd}\SW(f(u),u)}{\EE_{u\sim\Dd} \max_\sigma \SW(\sigma,u)}
\end{align*}

Then we constrain $u$ to be a cooperation game. For a normal form game $\Gamma_u$, we define $\Tilde{u}=(\Tilde{u}_i)_{i\in[n]}$ in which $\Tilde{u}_i = \frac{1}{n}\sum_{i=1}^n u_i$.
Then we have $\SW(\sigma,u) = \SW(\sigma,\Tilde{u})$, which means that constraining $u$ to be a cooperation game will induce the same social welfare.
Then
\begin{align*}
    & \inf_\Dd \frac{\max_{f\in\Ff^\mathrm{NE}_\mathrm{both}} \EE_{u\sim \Dd}\SW(f(u),u)}{\EE_{u\sim\Dd} \max_\sigma \SW(\sigma,u)}
    = \inf_\Dd \frac{\max_{f\in\Ff^\mathrm{NE}_\mathrm{both}} \EE_{u\sim \Dd}\SW(f(u),\Tilde{u})}{\EE_{u\sim\Dd} \max_\sigma \SW(\sigma,\Tilde{u})}
\end{align*}

Denote $f_0$ be the approximator that always outputs uniform strategy on $\Bb_i(\Tilde{u}_i)$ for player $i$. It's obvious that $f_0$ is both OPI and PPE because the operations from $u$ to $f_0(u)$ are all permutation-equivariant. Then,
\begin{align*}
    & \inf_\Dd \frac{\max_{f\in\Ff^\mathrm{NE}_\mathrm{both}} \EE_{u\sim \Dd}\SW(f(u),\Tilde{u})}{\EE_{u\sim\Dd} \max_\sigma \SW(\sigma,\Tilde{u})}
    \ge \inf_\Dd \frac{\EE_{u\sim \Dd}\SW(f_0(u),\Tilde{u})}{\EE_{u\sim\Dd} \max_\sigma \SW(\sigma,\Tilde{u})}
\end{align*}

Ignore the infimum and the expectation operator, consider $\frac{\SW(f_0(u),\Tilde{u})}{\max_\sigma \SW(\sigma,\Tilde{u})}$ for arbitrary $\Tilde{u}$, denote $b$ be the maximum element appeared in $\Tilde{u}$, then the denominator equals $Nb$. But for the numerator, for player $i$, no matter what action $a_i\in\Bb_i(\Tilde{u}_i)$ she chooses, she always has probability at least $\prod_{j\ne i} \frac{1}{|\Bb_j|}\ge \frac{1}{M^{N-1}}$ to achieve revenue $b$, therefore inducing $\SW(f_0(u),\Tilde{u})\ge\frac{Nb}{M^{N-1}}$.

Then, $\frac{\SW(f_0(u),\Tilde{u})}{\max_\sigma \SW(\sigma,\Tilde{u})}\ge \frac{1}{M^{N-1}}$, so as $\inf_\Dd \frac{\EE_{u\sim \Dd}\SW(f_0(u),\Tilde{u})}{\EE_{u\sim\Dd} \max_\sigma \SW(\sigma,\Tilde{u})}$, $\SWR({\Ff}^\NE_\mathrm{both})$ and $\SWR({\Ff}^\NE_\mathrm{OPI})$.

Above all, 
\begin{align*}
\left.
\begin{aligned}
    \SWR(\Ff^{\mathrm{NE}}_{\mathrm{OPI}},\Ff^\NE_\mathrm{general})
    \\
    \SWR(\Ff^{\mathrm{NE}}_{\mathrm{both}},\Ff^\NE_\mathrm{general})
\end{aligned}
\right\}= \frac{1}{M^{N-1}}
\end{align*}

\subsubsection{Proof of \cref{eq:NEG:ratio:PPE}}
We next prove the theorem with respect to $\Ff^\mathrm{NE}_\mathrm{PPE}$that
\begin{align*}
    \SWR(\Ff^{\mathrm{NE}}_{\mathrm{PPE}},\Ff^\NE_\mathrm{general}) \le \frac{1}{M}
\end{align*}

Consider a bimatrix game and $\Aa_i=[M]$ for $i\in[2]$.
$\forall a \in \Aa, i \in [2]$, define the payoff $\bar u$ as follows: 
\begin{equation*}
    \bar u_i(a) =     \begin{cases}
        1 &,\text{if }a_1 = a_2
        \\
        0 &,\text{otherwise}
    \end{cases}
\end{equation*}
Define $U\coloneqq\{u'|u'=\rho_1 \rho_2 \bar u,\rho_i \in\Gg_i\}$ and $\Dd$ as a uniform distribution on $U$. 
Easy to certify that $U = \{u'| u'=\rho_1\Bar{u}, \rho_1 \in \Gg_1\} = \{u'| u'=\rho_2\Bar{u}, \rho_2 \in \Gg_2\}$ and $\Dd$ is a permutation-invariant distribution.

Let $\Tilde{f}\in \Tilde\Ff^\NE_\mathrm{general}$ be the NE oracle that $\Tilde{f}(\bar u)_i = 1$ and for any $u'=\circ_{i} \rho_i \bar u\in U$, $\Tilde{f}(u')_i=\rho_i (1)$. 
Intuitively, the oracle will choose the action that will provide all players with revenue of $1$, leading to a social welfare of $2$. 

For a permutation $\varrho$ on $[M]$, let $P_\varrho \in \{0,1\}^{M\times M}$ be the corresponding permutation matrix. Denote $\Pp$ as the set of all permutation matrice.
As a result, $\forall u \in U, \forall \rho_1 \in \Gg_1, \rho_1 u = (P_{\rho_1} {u}_1, P_{\rho_1}{u}_2) =: P_{\rho_1}u$ and 
$\forall \rho_2 \in \Gg_2, \rho_2 u = ({u}_1P_{\rho_2}^T, {u}_2P_{\rho_2}^T) =: uP_{\rho_2}^T$.
Specially, we have $P_\varrho \Bar{u} P_\varrho^T = \Bar{u}$.
% By PPE of $f \in \Ff^\NE_\mathrm{PPE}$, we have 
% \begin{equation*}
%     % f(\bar u)_1 = \varrho^{(1)} f((\varrho^{(1)})^{-1} \bar u)_1 = \varrho^{(1)} f((\varrho^{(1)})^{-1} \varrho^{(1)} \varrho^{(2)} \bar u)_1 = \varrho^{(1)} f(\varrho^{(2)} \bar u)_1 = \varrho f(\varrho^{(2)} \bar u)_1.
%     f(\bar u)_1 
%     = P_\varrho f(P_\varrho^{-1}\Bar{u})_1
%     = P_\varrho f(P_\varrho^{-1}P_\varrho \Bar{u} P_\varrho^T)_1 
%     = P_\varrho f(\Bar{u} P_\varrho^T)_1.
% \end{equation*}
% Therefore $P_\varrho^{-1}f(\bar u)_1 = f(\Bar{u} P_\varrho^T)_1$.
% Furthermore, by PPE of $f$, we have 
% \begin{equation*}
%     f(\rho_1 \rho_2 \Bar{u})_1 = \rho_1 f(\rho_2 \Bar{u})_1 = \rho_1 f(\Bar{u}P_{\rho_2}^T)_1 = \rho_1 P_{\rho_2}^{-1}f(\bar u)_1 = \rho_1 \rho_2^{-1}f(\Bar{u})_1.
% \end{equation*}
% similarly, we also have 
% \begin{equation*}
%     f(\rho_1 \rho_2 \Bar{u})_2 = f(\rho_2 \rho_1 \Bar{u})_2 = \rho_2 \rho_1^{-1}f(\Bar{u})_2.
% \end{equation*}
For $f \in \Ff^\NE_\mathrm{PPE}$, Denote $f(\Bar{u})=\sigma=(\sigma_1,\sigma_2)$. 
For permutation $\varrho$ in $[M]$ and payoff $u' = P_\varrho\Bar{u} = \Bar{u}(P_\varrho^T)^{-1}$, by PPE of $f$, we have 
$
    f(u')_1 = f(P_\varrho\Bar{u})_1 =  P_\varrho \sigma_1 = \varrho \sigma_1,
$
and 
$
    f(u')_2 = f(\Bar{u}(P_\varrho^T)^{-1})_2 = (P_\varrho)^{-1} \sigma_2 = \varrho^{-1}\sigma_2.
$
Then we have 
\begin{equation*}
    \SW(f(u'),u') 
    = \sum_i (P_\varrho\Bar{u})_i(\varrho \sigma_1, \varrho^{-1}\sigma_2) 
    = \sum_i \bar{u}_i(\sigma_1, \varrho^{-1}\sigma_2) 
    = \sum_i (\Bar{u}P_\varrho^T)_i(\sigma_1,\sigma_2) 
    = \SW(f(\Bar{u}), \Bar{u}P_\varrho^T)
\end{equation*}

Therefore
\begin{align*}
    \EE_{u\sim\Dd} \SW(f(u),u) 
    &= \frac{1}{|U|}\sum_{u'\in U} \SW(f(u'),u') \\
    &= \frac{1}{|U|}\sum_{P_{\varrho}\in \Pp} \SW(f(\Bar{u}), \Bar{u}P_\varrho^T) \\
    &= \frac{1}{|U|}\sum_{u=\Bar{u}(P_\varrho^T) \in U} \SW(f(\Bar{u}), u) \\
    &= \frac{1}{|U|}\SW(\sigma,\sum_{u'\in U}u').
\end{align*}
Since $|U|=\frac{1}{M!}$ and $\sum_{u'\in U}u'$ is a tensor with all elements equal to $(M-1)!$. 
Thus $\EE_{u\sim\Dd} \SW(f(u),u)=\frac{2}{M}$ and 
\begin{align*}
    \SWR(\Ff^{\mathrm{NE}}_{\mathrm{PPE}},\Ff^\NE_\mathrm{general})\le\frac{1}{M}
\end{align*}

\subsubsection{Proof of \cref{eq:NEG:ratio:both:oracle}}

% Assume $|A|=n=3$, otherwise, we could add zero row and coloum to $u$ such that the maximum social welfare doesnot change.

% Let $\frac{1}{2}>\varepsilon>0$ and $u$ as follows,
Consider a $3\times 3$ game as follows, where $\epsilon \in (0, \frac{1}{2})$:
\begin{equation*}
    u = \begin{bmatrix}
    \bm{1,1} & 0,0 & 0,\frac{1}{2}+\varepsilon
    \\
    0,0 & \bm{1,1} & 0,\frac{1}{2}+\varepsilon
    \\
    \frac{1}{2}+\varepsilon,0 & \frac{1}{2}+\varepsilon,0 & \varepsilon,\varepsilon
    \end{bmatrix}
\end{equation*}

It is obvious that $\max_{\sigma^* \subseteq \text{NE}(\Gamma_u)}\text{SW}(\sigma^*,u) = 2$, and the corresponding strategy has been bolded. However, for NE oracles with both PPE and OPI, it can only output a unique NE with a pure strategy that induces utility $(\varepsilon,\varepsilon)$.

Let $\rho_1=\rho_2=(2,1,3)$, we have $\rho_1 \rho_2 u = u$. From the analysis above we know if $\fNE\in\widetilde\Ff^\NE_\mathrm{both}$ and $\fNE(u)=(\sigma_1,\sigma_2)$, then $\sigma_1(1)=\sigma_1(2)$, $\sigma_2(1)=\sigma_2(2)$. We integrate the first two actions of player $1$ and player $2$ into a new action that will choose randomly between the first two actions, then we form the utility matrix below:
\begin{equation*}
    \overline{u} = \begin{bmatrix}
    \frac{1}{2},\frac{1}{2} & 0,\frac{1}{2}+\varepsilon
    \\
    \frac{1}{2}+\varepsilon,0 & \bm{\varepsilon,\varepsilon}
    \end{bmatrix}
\end{equation*}

There is a unique NE in this Prisoner's Dilemma, which has been bolded. The game $\overline{u}$ is the same with the game $u$ under the assumption that $\sigma_1(1)=\sigma_1(2)$ and $\sigma_2(1)=\sigma_2(2)$ in $u$. Then $\max_{f\in\widetilde\Ff^\NE_\mathrm{both}}\text{SW}(f(u),u) = 2\varepsilon$.
Since $\varepsilon$ can be arbitrarily small, we have
$\mathrm{SWR}_{2,3}(\widetilde\Ff^\NE_\mathrm{both},\widetilde{\Ff}^\NE_\mathrm{general}) = 0$.
As a result, we have $\mathrm{SWR}_{N,M}(\widetilde\Ff^\NE_\mathrm{both},\widetilde{\Ff}^\NE_\mathrm{general}) = 0$ for all $N \ge 2$ and $M \ge 3$.

% -----
% \input{_appendix_experiment.tex}

\end{document}